\documentclass[a4paper,UKenglish,cleveref, autoref]{lipics-v2019}
\usepackage{amsmath,hyperref,lineno,amssymb,amsthm,mathtools,url,verbatim}
\usepackage{nccmath}
\usepackage{pgfplots}
\pgfplotsset{grid style={red}}
\usepackage{filecontents}
\usepackage{microtype}
\usepackage{float}
\usepackage{tikz}
\usetikzlibrary{matrix}
\usepackage{tikz-cd}
\usepackage{algorithmic, algorithm2e}
\usepackage[inline]{enumitem}
\usetikzlibrary{decorations.markings}

\title{ Paired compressed cover trees guarantee a near linear parametrized complexity for all $k$-nearest neighbors search in an arbitrary metric space }
\titlerunning{Paired trees and near linear time algorithms for $k$-nearest neighbor search}

\author{Yury Elkin}{Materials Innovation Factory and Computer Science department, University of Liverpool, UK}{yura.elkin@gmail.com}{}{}

\author{Vitaliy Kurlin}{Materials Innovation Factory and Computer Science department, University of Liverpool, UK}{vitaliy.kurlin@gmail.com}{}{}

\authorrunning{Y.Elkin et. al.}

\Copyright{Yury Elkin, Vitaliy Kurlin}



\category{}

\relatedversion{}
\nolinenumbers
\supplement{}

\funding{The authors were supported by the £3.5M EPSRC grant EP/R018472/1 (2018-2023) 
}

\EventEditors{}
\EventNoEds{0}
\EventLongTitle{}
\EventShortTitle{}
\EventAcronym{}
\EventYear{}
\EventDate{}
\EventLocation{}
\EventLogo{}
\SeriesVolume{}
\ArticleNo{}
\nolinenumbers 

\usepackage{verbatim,amsfonts,amsmath,amssymb}
\usepackage{hyperref}
\usepackage{algorithmic}

\usepackage{tikz}
\usetikzlibrary{matrix}
\usetikzlibrary{arrows,positioning,shapes.geometric}
\usetikzlibrary{calc,trees,positioning,arrows,chains,shapes.geometric,%
    decorations.pathreplacing,decorations.pathmorphing,shapes,%
   shapes.symbols, trees,backgrounds}
      \tikzstyle{blockzm1} = [rectangle, draw, fill=white, 
    text width=1.5em, text centered, rounded corners, minimum height=1.5em, minimum width = 1.5em]
    
      \tikzstyle{blockzm2} = [rectangle, draw, fill=white, 
    text width=1.0em, text centered, rounded corners, minimum height=1.0em, minimum width = 1.0em]
    
      \tikzstyle{blockzm1r} = [rectangle, draw, fill=red!70, 
    text width=1.5em, text centered, rounded corners, minimum height=1.5em, minimum width = 1.5em]
      \tikzstyle{blockzm1o} = [rectangle, draw, fill=orange!80, 
    text width=1.5em, text centered, rounded corners, minimum height=1.5em, minimum width = 1.5em]
      \tikzstyle{blockzm1y} = [rectangle, draw, fill=yellow!50, 
    text width=1.5em, text centered, rounded corners, minimum height=1.5em, minimum width = 1.5em]
      \tikzstyle{blockzm1g} = [rectangle, draw, fill=green!60, 
    text width=1.5em, text centered, rounded corners, minimum height=1.5em, minimum width = 1.5em]
        \tikzstyle{blockzm1p} = [rectangle, draw, fill=blue!50, 
    text width=1.5em, text centered, rounded corners, minimum height=1.5em, minimum width = 1.5em]
    
   \tikzstyle{blockz} = [rectangle, draw, fill=white, 
    text width=1.75em, text centered, rounded corners, minimum height=1.75em, minimum width = 1.75em]
     \tikzstyle{blockzflexi} = [rectangle, draw, fill=white, 
    text centered, rounded corners]
    \tikzstyle{blockz2} = [rectangle, draw, fill=white, 
    text width=2em, text centered, rounded corners, minimum height=2em, minimum width = 2em]
        \tikzstyle{blockzL} = [rectangle, draw, fill=white, 
    text width=4em, text centered, rounded corners, minimum height=3em, minimum width = 4em]
    
       \tikzstyle{vertex}=[circle,fill=black!25,minimum size=20pt,inner sep=0pt]
\tikzstyle{selected vertex} = [vertex, fill=red!24]
\tikzstyle{edge} = [draw,thick,-]
\tikzstyle{weight} = [font=\small]
\tikzstyle{selected edge} = [draw,line width=5pt,-,red!50]
\tikzstyle{ignored edge} = [draw,line width=5pt,-,black!20]
    
\tikzset{
  treenode/.style = {align=center, inner sep=0pt, text centered,
    font=\sffamily},
  arn_n/.style = {treenode, circle, white, font=\sffamily\bfseries, draw=black,
    fill=black, text width=1.5em},
  arn_r/.style = {treenode, circle, red, draw=red, 
    text width=1.5em, very thick},
  arn_x/.style = {treenode, rectangle, draw=black,
    minimum width=0.5em, minimum height=0.5em}
}

\usepackage{thmtools}
\usepackage{thm-restate}

\usepackage{tikz-cd}
\usepackage{verbatim}
\usepackage{subcaption}
\numberwithin{algocf}{section}

\theoremstyle{definition}
\newtheorem{dfn}[algocf]{Definition}
\newtheorem{rem}[algocf]{Remark}
\newtheorem{pro}[algocf]{Problem}
\newtheorem{thmm}[algocf]{Theorem}

\theoremstyle{definition}
\newtheorem{exa}[algocf]{Example}
\theoremstyle{definition}

\newtheorem{cexa}[algocf]{Counterexample}
\theoremstyle{plain}

\newtheorem{cor}[algocf]{Corollary}
\newtheorem{lem}[algocf]{Lemma}

\usepackage{mathtools}
\DeclarePairedDelimiter\ceil{\lceil}{\rceil}

\newcommand{\Es}{\mathcal{E}}

\newcommand{\Child}{\mathrm{Children}}
\newcommand{\Desc}{\mathrm{Descendants}}
\newcommand{\nxt}{\mathrm{Next}}

\newcommand{\R}{\mathbb{R}}
\newcommand{\Z}{\mathbb{Z}}

\newcommand{\C}{\mathcal{C}}

\newcommand{\T}{\mathcal{T}}
\newcommand{\Sd}{\mathcal{S}}

\newcommand{\be}{\beta}

\newcommand{\De}{\Delta}
\newcommand{\ep}{\epsilon}

\newcommand{\la}{\lambda}

\newcommand{\NN}{\mathrm{NN}}

\newcommand{\rad}{\mathrm{diam}}

\newcommand{\bs}{\hfill $\blacksquare$}

\begin{document}
\maketitle

\begin{abstract}
This paper studies the important problem of finding all $k$-nearest neighbors to points of a query set $Q$ in another reference set $R$ within any metric space. 
Our previous work defined compressed cover trees and corrected the key arguments in several past papers for challenging datasets. 
In 2009 Ram, Lee, March, and Gray attempted to improve the time complexity by using pairs of cover trees on the query and reference sets.  
In 2015 Curtin with the above co-authors used extra parameters to finally prove a time complexity for $k=1$.
The current work fills all previous gaps and improves the nearest neighbor search based on pairs of new compressed cover trees. 
The novel imbalance parameter of paired trees allowed us to prove a better time complexity for any number of neighbors $k\geq 1$. 
\end{abstract}

\section{All $k$-nearest neighbors problem and overview of new results}
\label{sec:intro}

The nearest neighbor search was one of the first data-driven problems \cite{cover1967nearest}.
Briefly, the problem is to find all $k\geq 1$ nearest neighbors in a reference set $R$ for all points from a query set $Q$.
Both sets live in an ambient space $X$ with a distance $d$ satisfying all metric axioms.
For example, $X=\R^m$ with the Euclidean metric, a query set $Q$ can be a single point.
\medskip

The naive approach to find first nearest neighbors of points from $Q$ within $R$ is proportional to the product $|Q|\cdot|R|$ of the sizes of $Q, R$.
The aim in nearest neighbor search is to reduce the brute-force complexity $O(|Q|\cdot|R|)$ to a near linear time in the maximum size of $Q,R$.
\medskip
 
The \emph{exact} $k$-nearest neighbor problem asks for exact (true) $k$-nearest neighbors of every query point $q$. 
The \emph{probabilistic} neighbor problem \cite{manocha2007empirical} aims to find exact $k$-nearest neighbors with a given probability. 
The \emph{approximate} neighbor problem \cite{arya1993approximate}, \cite{krauthgamer2004navigating},\cite{andoni2018approximate},\cite{wang2021comprehensive} aims to find for every query point $q \in Q$ its approximate neighbor $r\in R$ satisfying $d(q,r) \leq (1+\epsilon)d(q,\NN(q))$, where $\ep>0$ is fixed and $\NN(q)$ is the exact first nearest neighbor of $q$.
\medskip

Arguably the first approach to the neighbor search was developed in 1974 by using a \emph{quadtree} \cite{finkel1974quad}, which hierarchically indexed a reference set $R\subset\R^2$.
In higher dimensions, a KD-tree \cite{bentley1975multidimensional} smarter subdivided a subset of $R\subset\R^m$ at every level into two subsets instead of $2^m$ subsets. 
\cite[Section~1]{elkin2021new} described further developments since 1970s in detail.
Section~\ref{sec:review} will technically compare the most recent results after formalizing the key concepts below.
\medskip

We formally define the \emph{$k$-nearest neighbor set} $\NN_k(q)$, which might contain several points in a singular case when these points are at the exactly same distance from a query point $q$.

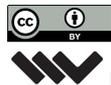
\begin{dfn}[$k$-nearest neighbor set $\NN_k$]
\label{dfn:kNearestNeighbor}
Let $Q,R$ be finite subsets of a metric space. 
For any points $q \in Q$ and $r\in R$, the \emph{neighborhood} 
$N(q; u) = \{p \in R \mid d(q,p) \leq d(q,u)\}$ consists of all points that are non-strictly closer to $q$ than $u$ is close to $q$. 
For any integer $k\geq 1$, the \emph{$k$-nearest neighbor set} $\NN_k(q)$ consists of all points $u\in R$ such that $|N(q;u)|\geq k$ and any other $v\in R$ with $d(v,q) > d(u,q)$ has a larger neighborhood of size $|N(q;v)| > k$. 
\bs
\end{dfn}

For $Q = R = \{0,1,2,3\}$, the nearest neighbor sets of $0$ are
$\NN_1(0) = \{1\}$, $\NN_2(0) = \{2\}$, $\NN_3(0) = \{3\}$.
Due to the neighborhoods $N(1;0)=\{0,1,2\}=N(1;2)$, both sets $\NN_1(1) = \{0,2\}=\NN_2(1)$ consist of two points $0,2$ at equal distance to $1$.
The 3-nearest neighbor set $\NN_3(1) = \{3\}$ is a single point.
Because of a potential ambiguity of an exact $k$-nearest neighbor, Problem \ref{pro:knn} below allows any neighbors within a set $\NN_k(q)$.
In the above example, $0$ can be chosen as a 1st neighbor of $1$, then $2$ as a 2nd neighbor of $1$, or vice versa.

\begin{pro}[all $k$-nearest neighbors search]
\label{pro:knn}
For any finite sets $Q,R$ in a metric space and for any integer $k\geq 1$, design an algorithm to exactly find a distinct points $p_i\in\NN_{i}(q) \subseteq R$ for all $i = 1,..., k$ and all points $q\in Q$ such that the total complexity is near linear in $n=\max\{|Q|,|R|\}$ with hidden constants depending on structures of the sets $Q,R$. 
\bs
\end{pro}

We solve Problem \ref{pro:knn} by using new compressed cover trees on both sets $Q,R$.
\cite[Section~2]{beygelzimer2006cover} introduced a first version of a cover tree, which implicitly repeats every data point at infinitely many levels, see a comparison of two trees on the same $R$ in
Fig.~\ref{fig:implicitcompressed}.
The new tree $\T(R)$ in the right hand side of Fig.~\ref{fig:implicitcompressed} has nodes at levels $-1,0,1,2$, so its height is $|H(\T(R))|=4$.
\medskip

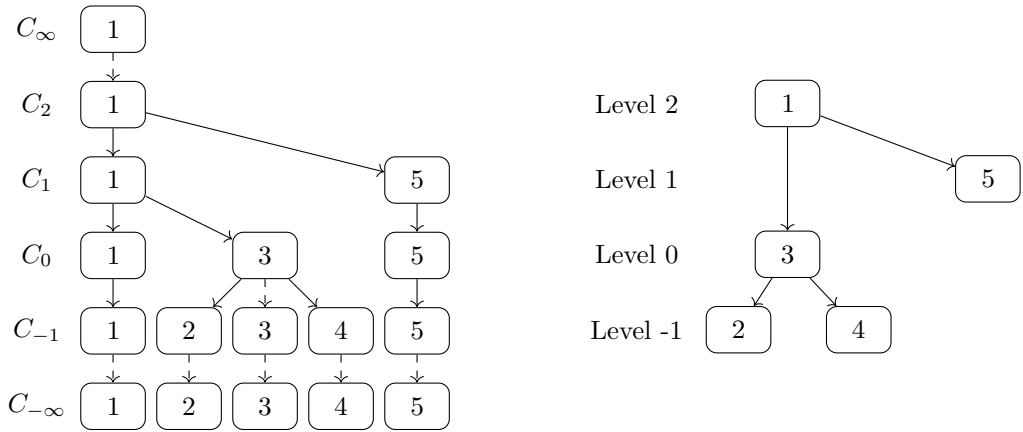
\begin{figure}
\centering
\begin{tikzpicture}[align=center, node distance = 1.0cm, scale = 0.45]
    
	\node (scaleminftext) {$C_{-\infty}$};
	\node [blockz, right of =scaleminftext] (scaleminf1) {1};
	\node [blockz, right of =scaleminf1] (scaleminf2) {2};
	\node [blockz, right of =scaleminf2] (scaleminf3) {3};
	\node [blockz, right of =scaleminf3] (scaleminf4) {4};
	\node [blockz, right of =scaleminf4] (scaleminf5) {5};

	\node [above of =scaleminftext](scalem1text) {$C_{-1}$};
	\node [blockz, above of =scaleminf1] (scalem11) {1};
	\node [blockz, above of =scaleminf2] (scalem12) {2};
	\node [blockz, above of =scaleminf3] (scalem13) {3};
	\node [blockz, above of =scaleminf4] (scalem14) {4};
	\node [blockz, above of =scaleminf5] (scalem15) {5};
	
	\node [above of =scalem1text](scale0text) {$C_{0}$};
	\node [blockz, above of =scalem11] (scale01) {1};
	\node [blockz, above of =scalem13] (scale03) {3};
	\node [blockz, above of =scalem15] (scale05) {5};	
	
	\node [above of =scale0text](scale1text) {$C_{1}$};
	\node [blockz, above of =scale01] (scale11) {1};
	\node [blockz, above of =scale05] (scale15) {5};
	
	\node [above of =scale1text](scale2text) {$C_{2}$};
	\node [blockz, above of =scale11] (scale21) {1};

	\node [above of =scale2text](scaleinftext) {$C_{\infty}$};
	\node [blockz, above of =scale21] (scaleinf1) {1};
    
      \draw[dashed,->] (scalem11) -> (scaleminf1);
      \draw[dashed,->] (scalem12) -> (scaleminf2);
      \draw[dashed,->] (scalem13) -> (scaleminf3);
      \draw[dashed,->] (scalem14) -> (scaleminf4);
      \draw[dashed,->] (scalem15) -> (scaleminf5);

	  \draw[->] (scale01) -> (scalem11);

	  \draw[->] (scale03) -> (scalem12); 
     \draw[dashed, ->] (scale03) -> (scalem13);

	  \draw[->] (scale05) -> (scalem15);

	  \draw[->] (scale03) -> (scalem14);
	  
	  \draw[dashed,->] (scaleinf1) -> (scale21);
	  \draw[->] (scale21) -> (scale11);
	  \draw[->] (scale15) -> (scale05);
	  
	  \draw[->] (scale11) -> (scale01);
	  \draw[->] (scale11) -> (scale03);
	  \draw[->] (scale21) -> (scale15);
      
\end{tikzpicture} 
\hspace{1.5cm}
\begin{tikzpicture}[align=center, node distance = 1.0cm, scale = 0.45 ,baseline={(0,-4.33)}]
    

	\node (scalem2text) {Level 2};
	\node[below of =scalem2text] (scalem1text) {Level 1};
	\node[below of =scalem1text] (scalem0text) {Level 0};
	\node[below of =scalem0text] (scalemm1text) {Level -1};
	\node [blockz,  right=25pt of scalem2text ] (node1) {1};
	\node [blockz,  right=100pt of scalem1text] (node5) {5};
	\node [blockz,  right=25pt of scalem0text] (node3) {3};
	\node [blockz,  right=5pt of scalemm1text] (node2) {2};
	\node [blockz,  right=50pt of scalemm1text] (node4) {4};

	  \draw[->] (node1) -> (node5);
	  \draw[->] (node1) -> (node3);
	  \draw[->] (node3) -> (node2);
	  \draw[->] (node3) -> (node4);

	
	
	

    

	  
	  
	  
      
\end{tikzpicture} 
  \caption{\textbf{Left:} an implicit cover tree for the finite set of reference points $R = \{1,2,3,4,5\}$ with the Euclidean distance $d(x,y) = |x-y|$, initially introduced in \cite[Section~2]{beygelzimer2006cover}. 
  \textbf{Right:} a new compressed cover tree in Definition~\ref{dfn:cover_tree_compressed} corrects the past complexity results for neighbor search.}
  \label{fig:implicitcompressed}
\end{figure}

The height parameter $|H(\T(R))|$ will play an important role to correctly estimate the time complexity, see Definition~\ref{dfn:depth}.
The new concept of the imbalance $I(\T(R),\T(Q))$ in Definition~\ref{dfn:imbalance} is the final parameter needed to express the search time complexities.
\medskip

Section~\ref{sec:review} discusses challenging cases that were overlooked in the past.
Sections~\ref{sec:cover_tree} and~\ref{sec:distinctive_descendant_set} introduce key concepts and prove auxiliary results.
Section~\ref{sec:generaldualtreetheory} estimates the time complexity for a traversal of paired trees in Theorem~\ref{thm:paired_cover_tree_traversal_bound}.
Section~\ref{sec:CorrectDualTreeKNN} describes Algorithm~\ref{alg:paired_tree_knn} whose correctness and complexity are proved in Theorems~\ref{thm:paired_tree_knn_correctness},~\ref{thm:paired_tree_knn_time} solving Problem~\ref{pro:knn} for $k\geq 1$.

\section{Review of past work and challenges in the $k$-nearest neighbor search}
\label{sec:review}

First we review the past searches by using a single tree on a reference set $R$.
Then we consider the faster searchers by using a paired trees (a dual tree approach) on both $Q,R$. 
\medskip

\textbf{Search on a single tree}.
The highly-cited paper \cite{beygelzimer2006cover} by Beygelzimer, Kakade, Langford in 2006 introduced the first concept of a cover tree to prove a worst-case time complexity for a 1st nearest neighbor search ($k=1$).
This tree was motivated by the earlier navigating nets \cite{krauthgamer2004navigating} and is called an \emph{implicit cover tree} to distinguish other versions in Fig.~\ref{fig:implicitcompressed} and~\ref{fig:tripleexample}.
\medskip

\cite[Theorem~6]{beygelzimer2006cover} implied that an implicit cover tree $T$ on a reference set $R$ can be built in time $O(c^6|R|\log|R|)$, where $|R|$ is the size (number of points) of $R$, $c$ is an expansion constant depending a structure of $R$, see Definition~\ref{dfn:expansion_constant}. 
Assuming that $T$ is already built, \cite[Theorem~5]{beygelzimer2006cover} expressed the time complexity to find a 1st nearest neighbor in $R$ for all points $q\in Q$ as $O(c^{12}|Q|\log|R|)$. 
However, both statements relied on \cite[Lemma~4.3]{beygelzimer2006cover} showing that a depth of a single point in a tree is $O(c^2\log|R|)$.
In 2015 \cite[Section~5.3]{curtin2015improving} pointed out \cite[Lemma~4.3]{beygelzimer2006cover} cannot be used to estimate the number of internal recursions in \cite[Theorem~5]{beygelzimer2006cover}. In \cite[section~3]{elkin2021new} it was shown that the original time complexity theorem for the cover tree construction algorithm \cite[Theorem~6]{beygelzimer2006cover} had similar issues. 
Otherwise, the past argument guarantees the total complexity $O(c^{12}|R|^2)$ for $Q=R$ as for a brute-force search for $k=1$.
In 2021 \cite[section~3]{elkin2021new} described typical sets $R$ that indeed require more careful estimates. 
\medskip

The above facts motivated a new concept of a compressed cover tree to clarify the worst-case complexity for all reference sets $R$, see Definition~\ref{dfn:cover_tree_compressed}, which 
\cite[Theorem~6.5]{elkin2021new} showed that, after constructing a compressed cover tree $\T(R)$ on $R$, we can find all $k$-nearest neighbors of $q$ in $R$ in time $O(c^{10}k\log (k) \cdot |H(\T(R))|)$  for any $k\geq 1$, where $|H(\T(R))|$ is the height of $\T(R)$, see Definition \ref{dfn:depth}.
The new height parameter $|H(\T(R))|$ has the upper bound $\log_2(\Delta)$, where the aspect ratio $\Delta$ is the diameter of $R$ divided by the minimum distance between points in $R$. 
Another approach uses paired trees on both sets $Q,R$.
\medskip

\textbf{Search on paired trees}. 
In 2001 Gray and Moore \cite{gray2001n} considered two KD-trees on the reference and query sets $Q,R$.
The approach was often called a dual-tree search.
Since the trees on $Q,R$ are independent, there is no real duality between them, so we call such trees \emph{paired}.
In 2006 \cite{beygelzimer2006coverExtend} proposed to consider paired cover trees.
\cite[Theorem~3.7]{beygelzimer2006coverExtend} claimed that \cite[Algorithm~5]{beygelzimer2006coverExtend} finds all 1st nearest-neighbors for a query set $Q$ in time $O(c^{16}\max\{|Q|,|R|\})$.
In the case $Q=R$, \cite[Algorithm~5]{beygelzimer2006coverExtend}
seems to return only the self-neighbor $q$ for any point $q\in R$.
A sketched proof of \cite[Theorem~3.7]{beygelzimer2006coverExtend} claimed similarities with the argument in \cite[Theorem~5]{beygelzimer2006cover} without details that could clarify the concerns in the single tree case.
\medskip

In 2009 \cite[Theorem~3.1]{ram2009linear} revisited the time complexity for all 1st nearest neighbors and claimed the upper bound $O(c(R)^{12}c(Q)^{4\kappa}\max\{|Q|,|R|\})$, where $c(Q),c(R)$ are expansion constants of the query set $Q$ and reference set $R$.
The degree of bichromaticity $\kappa$ is a parameter of both sets $Q,R$, see \cite[Definition~3.1]{ram2009linear}.
We have found the following issues.
\medskip

First, Counterexample~\ref{cexa:dualtreecode} shows that \cite[Algorithm~1]{ram2009linear} for $Q=R$ returns for any query point $q\in Q$ the same point $q$ as its first neighbor. 
Second, Remark~\ref{rem:kappa} explains several possible interpretations of \cite[Definition~3.1]{ram2009linear} for the parameter $\kappa$. 
Third, \cite[Theorem~3.1]{ram2009linear} similarly to \cite[Theorem~5]{beygelzimer2006cover} relied on the same estimate of recursions in the proof of \cite[Lemma~4.3]{beygelzimer2006cover}.
Counterexample~\ref{cexa:dualtreeproof} explains step-by-step why the time complexity of \cite[Algorithm~1]{ram2009linear} requires a clearer definition of $\kappa$ and more detailed justifications.
\medskip

In 2015 Curtin with the authors above  \cite{curtin2015plug} introduced other parameters:
the imbalance $I_t$ in \cite[Definition~3]{curtin2015plug} and $\theta$ in \cite[Definition~4]{curtin2015plug}.
These parameters measured extra recursions that occurred due to possible imbalances in trees built on $Q,R$, which was missed in the past. 
\cite[Theorem~2]{curtin2015plug} shows that, for constructed cover trees on a query set $Q$ and a reference set $R$, Problem~\ref{pro:knn} for $k=1$ (only 1st nearest neighbors) can be solved in time 
$$O\Big(c^{O(1)}\big(|R| + |Q| + I_t + \theta\big)\Big).\eqno{(*)}$$ 

\textbf{The main novelty} of this work is Theorem~\ref{thm:paired_tree_knn_time} proving a better time complexity in Problem~\ref{pro:knn} for any $k\geq 1$ (all $k$-nearest neighbors), which is a bit simplified here as
$$O\Big( c^{O(1)}k\log(k) \cdot\big( |H(\T(R))| + |Q| + I(\T(R),\T(Q)) \big)\Big).\eqno{(**)}$$
The new single imbalance $I$ of both compressed cover trees on $Q,R$, roughly combines the above parameters $I_t$ and $\theta$ from \cite{curtin2015plug}, see Definition~\ref{dfn:imbalance}.
The height $|H(\T(R))|$ is bounded above by $\log_2(\Delta)$, where $\De$ is the aspect ratio of the reference set $R$, see Definition~\ref{dfn:radius+d_min}. 
\medskip

In practice, we often need to quickly find nearest neighbors for a small number of data objects within a much larger dataset.
So the reference set $R$ is usually larger than a query set $Q$, which might be a single point.
Hence $(**)$ improves the previous time complexity in (*) 
from \cite[Theorem~2]{curtin2015plug} by replacing the large term $|R|$ by the height $|H(\T(R))|$.
\medskip

Applying the upper bound $I(\T(R),\T(Q))\leq |Q| \cdot |H(\T(R))|$ from
Lemma~\ref{lem:imbalanceproperties}, the complexity in $(**)$ reduces to \cite[Theorem~6.5]{elkin2021new}.
Hence the new imbalance $I(\T(R),\T(Q))$ for paired trees improves the recent approach \cite{elkin2021new} to Problem~\ref{pro:knn} based on a single tree.
\medskip

Finally, in the case of a full batch search for $Q=R$, Lemma \ref{lem:perfect_tree_imbalance} shows that if $\T(R)$ is a balanced tree, then $I(\T(R),\T(R)) = O(|R|)=O(|Q|)$.
Hence in all cases the new parameterized complexity is near linear in the maximum size $\max\{|R|,|Q|\}$ and any $k\geq 1$.

\begin{figure}[h]
\centering
   \begin{subfigure}{.30\textwidth}
  \centering
  \begin{tikzpicture}[align=center, node distance = 1.0cm, scale = 0.45]

	\node (scalem2text) {Level $i$};
	\node[below of =scalem2text] (scalem1text) {Level $i-1$};
	\node[below of =scalem1text] (scalem0text) {Level -1};
	\node [blockz,  right=30pt of scalem2text ] (node1) {$2^{i}$};
    \node [blockz,  right=25pt of scalem1text ] (node2) {1};
	\node [blockz,  right=32pt of scalem0text ] (node3) {0};
	

	  \draw[->] (node1) -> (node2);
	   \draw[->] (node2) -> (node3);

\end{tikzpicture}
  \label{fig:cover_tree_variant_one}
\end{subfigure}
\begin{subfigure}{.30\textwidth}
  \centering
  \begin{tikzpicture}[align=center, node distance = 1.0cm, scale = 0.45]

	\node (scalem2text) {Level $i$};
	\node[below of =scalem2text] (scalem1text) {Level $i-1$};
	\node[below of =scalem1text] (scalem0text) {Level -1};
	\node [blockz,  right=30pt of scalem2text ] (node1) {$2^{i}$};
    \node [blockz,  right=25pt of scalem1text ] (node2) {0};
	\node [blockz,  right=32pt of scalem0text ] (node3) {1};
	

	  \draw[->] (node1) -> (node2);
	   \draw[->] (node2) -> (node3);

\end{tikzpicture}
  \label{fig:cover_tree_variant_two}
\end{subfigure}
\begin{subfigure}{.30\textwidth}
  \centering
  \begin{tikzpicture}[align=center, node distance = 1.0cm, scale = 0.45, baseline={(0,-1.33)}]

	\node (scalem2text) {Level $i$};
	\node[below of =scalem2text] (scalem1text) {Level $i-1$};
	\node[below of =scalem1text] (scalem0text) {Level -1};
	\node [blockz,  right=30pt of scalem2text ] (node1) {0};
    \node [blockz,  right=50pt of scalem1text ] (node2) {$2^{i}$};
	\node [blockz,  right=32pt of scalem0text ] (node3) {1};
	

	  \draw[->] (node1) -> (node2);
	   \draw[->] (node1) -> (node3);

\end{tikzpicture}
  \label{fig:cover_tree_variant_three}
\end{subfigure}
\caption{For any integer $i\geq 2$, the reference set $R = \{0,1,2^{i}\}$ with Euclidean metric has at least three compressed cover trees $\T(R)$ satisfying all conditions of Definition~\ref{dfn:cover_tree_compressed}. }
\label{fig:cover_tree_easy_example}
\end{figure}
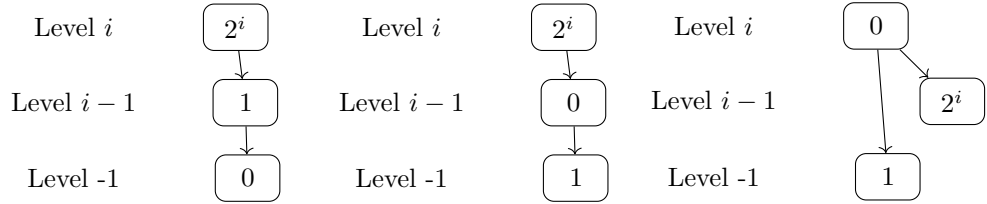

\section{A compressed cover tree and other tools for $k$-nearest neighbors}
\label{sec:cover_tree}


\begin{dfn}[a compressed cover tree $\T(R)$, {\cite[Definition~2.1]{elkin2021new}}]
\label{dfn:cover_tree_compressed}
Let $R$ be a finite set in an ambient space $X$ with a metric $d$. 
 \emph{A compressed cover tree} $\T(R)$ has the vertex set $R$ with a root $r \in R$ and a \emph{level} function $l : R \rightarrow \Z$ satisfying the conditions below.
\medskip
 
\noindent
(\ref{dfn:cover_tree_compressed}a)
\emph{Root condition} :  
the level of a root node $r$ is $l(r) \geq 1 +  \max_{p \in R-r}l(p)$.
\medskip

\noindent
(\ref{dfn:cover_tree_compressed}b)
\emph{Covering condition} : 
every non-root node $q \in R$ has a unique \emph{parent} $p$ and a level $l(q)<l(p)$ so that $d(q,p) \leq 2^{l(q)+1}$; 
then $p$ has a single link to its \emph{child} node $q$ in $\T(R)$. 
\medskip

\noindent
(\ref{dfn:cover_tree_compressed}c)
\emph{Separation condition} : 
for any $i \in \Z$, the \emph{cover set} $C_i = \{p \in R \mid l(p) \geq i\}$ has a minimum separation distance
$d_{\min}(C_i)
     = \min\limits_{p \in C_{i}}\min\limits_{q \in C_{i}-p} d(p,q) > 2^{i}$.
\medskip

\noindent
Since there is a 1-1 correspondence between all reference points in $R$ and all nodes of the tree $\T(R)$, the same notation $p$ can refer to a point of the reference set $R$ or to a node of a tree $\T(R)$.  
Set $l_{\max} = 1 +  \max\limits_{p \in R-r}l(p) $ and $l_{\min} = \min\limits_{p \in R}l(p)$.
For a node $p\in\T(R)$, let $\Child(p)$ denote the set of all children of $p$, including $p$ itself, which will be convenient later.
\medskip
 
For any node $p \in\T(R)$, define the \emph{node-to-root} path as a unique sequence of nodes $w_0,\dots,w_m$ such that $w_0 = p$, $w_m$ is the root and $w_{j+1}$ is the parent of $w_{j}$ for all $j=0,...,m-1$. 
A node $q \in\T(R)$ is called a \emph{descendant} of another node $p$ if $p$ belongs to the node-to-root path of $q$. 
A node $p$ is an \emph{ancestor} of $q$ if $q$ belongs to the node-to-root path of $p$. 
The set of all descendants of a node $p$ is denoted by $\Desc(p)$ and includes the node $p$ itself. 
\bs
\end{dfn} 

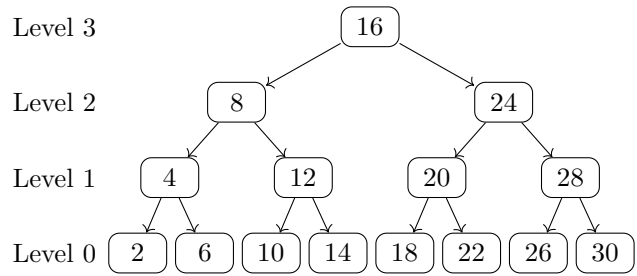
\begin{figure}[h]
\centering
\begin{tikzpicture}[align=center, node distance = 1.0cm, scale = 0.45]

	\node (scale3) {Level $3$};
	\node[below of =scale3] (scale2) {Level 2};
	\node[below of =scale2] (scale1) {Level 1};
	\node[below of =scale1] (scale0) {Level $0$};
	\node [blockzm1,  right=1pt of scale0 ] (node1) {2};
	\node [blockzm1,  right=26pt of scale0 ] (node3) {6};
	\node [blockzm1,  right=51pt of scale0 ] (node5) {10};
	\node [blockzm1,  right=76pt of scale0 ] (node7) {14};
	\node [blockzm1,  right=101pt of scale0 ] (node9) {18};
	\node [blockzm1,  right=126pt of scale0 ] (node11) {22};
	\node [blockzm1,  right=151pt of scale0 ] (node13) {26};
	\node [blockzm1,  right=176pt of scale0 ] (node15) {30};
		\node [blockzm1,  right=13pt of scale1 ] (node2) {4};
		\node [blockzm1,  right=63pt of scale1 ] (node6) {12};
		\node [blockzm1,  right=113pt of scale1 ] (node10) {20};
		\node [blockzm1,  right=163pt of scale1 ] (node14) {28};

		\node [blockzm1,  right=38pt of scale2 ] (node4) {8};
		\node [blockzm1,  right=138pt of scale2 ] (node12) {24};
		\node [blockzm1,  right=88pt of scale3 ] (node8) {16};
	

	
    \draw[->] (node2) -> (node1);
    \draw[->] (node2) -> (node3);
     \draw[->] (node6) -> (node5);
    \draw[->] (node6) -> (node7);
     \draw[->] (node10) -> (node9);
    \draw[->] (node10) -> (node11);
     \draw[->] (node14) -> (node13);
    \draw[->] (node14) -> (node15);

    \draw[->] (node8) -> (node4);
    \draw[->] (node8) -> (node12);

      \draw[->] (node4) -> (node2);
     \draw[->] (node4) -> (node6);
         
         \draw[->] (node12) -> (node10);
     \draw[->] (node12) -> (node14);

\end{tikzpicture}
\caption{Compressed cover tree $\T(R)$ built on set $R$ defined in Example \ref{exa:cover_tree_big} with root $16$. }
\label{fig:cover_tree_big}
\end{figure}

\begin{exa}[tree $\T(R)$ in Fig.~\ref{fig:cover_tree_big}]
\label{exa:cover_tree_big}
Let $\R$ be the real line with Euclidean metric $d = |x-y|$.
Set $R = \{2,4,6,...,30\}\subset\R$. 
Figure~\ref{fig:cover_tree_big} shows a compressed cover tree on the set $R$ with root $16$. 
The cover sets of $\T(R)$ are $ C_{0} = \{2,4,6,...,30\}$, $C_1 = \{4,8,12,16,20,24,28\}$, $C_{2} = \{8,16,24\}$ and $C_{3} = \{16\}$.
We check the conditions of Definition \ref{dfn:cover_tree_compressed} below.
 \begin{itemize}
\item Root condition (\ref{dfn:cover_tree_compressed}a): $\max_{p \in R \setminus \{16\}}d(p, 16) = 14$ and $\ceil{\log_2(14)} - 1= 3$, so $l(16) = 2$.
    \item Covering condition (\ref{dfn:cover_tree_compressed}b) : for any $i \in {0,1,2,3}$, let $p_i$ be any point with level $l(p_i) = i$. Then  $d(p_{0}, p_{1}) = 2 \leq 2^{1}$. $d(p_1, p_2) = 4 \leq 2^{2}$ and  $d(p_2, p_3) = 8 \leq 2^{3}$.
     \item  Condition (\ref{dfn:cover_tree_compressed}c) : $d_{\min}(C_{0}) = 2 > 1 = 2^{0}$, $d_{\min}(C_{1}) = 4 > 2 = 2^{1}, d_{\min}(C_{2}) = 8 > 2^{2}$. 
 \bs
 \end{itemize}
\end{exa}

A cover tree was defined in \cite[Section~2]{beygelzimer2006cover} as a tree version of a navigating net \cite[Section ~ 2.1]{krauthgamer2004navigating}. 
All nodes at a fixed level $i \in \Z\cup \{\pm\infty\}$ form the cover set $C_i$ from Definition~\ref{dfn:cover_tree_compressed}, which can have nodes at different levels in a compressed cover tree. 
Any point $p \in C_i$ had a single parent in the set $C_{i+1}$, which satisfied conditions (\ref{dfn:cover_tree_compressed}b,c). 
The original tree was called an implicit representation of a cover tree in
\cite[Section ~ 2]{beygelzimer2006cover}.
This tree in Figure \ref{fig:tripleexample} (left) contains infinite repetitions of every point $p\in R$ and will be called an \emph{implicit cover tree}.
\medskip

For practical implementations, the authors of \cite{beygelzimer2006cover} had to use another version of a cover tree that they named as an explicit representation of a cover tree and we call an \emph{explicit cover tree}. 
Here is the full defining quote at the end of \cite[Section~2]{beygelzimer2006cover}: "The explicit representation of the tree coalesces all nodes in which the only child is a self-child". 
In an explicit cover tree, if a subpath of every node-to-root path consists of all identical nodes without other children, all these identical nodes collapse to a single node, see Figure \ref{fig:tripleexample} (middle). 
\medskip

Since an explicit cover tree still repeats points, Definition~\ref{dfn:cover_tree_compressed} is well-motivated by including every point exactly once to save memory and simplify later algorithms, see Figure \ref{fig:tripleexample} (right).

\begin{figure}
    \centering
   \begin{tikzpicture}[align=center, node distance = 1.0cm, scale = 0.45]

    \node (scaleinf) {$l = \infty$};
    \node [below = 12.5pt of scaleinf](scalemidinf) {};
     \node[below=25pt of scaleinf] (scale6) {$l = 5$};
    \node[below of =scale6] (scale5) {$l = 4$};
    \node[below of =scale5] (scale4) {$l = 3$};
	\node[below of =scale4] (scale3) {$l = 2$};
	\node[below of =scale3] (scale2) {$l = 1$};
	\node[below=25pt of scale2] (scale1) {$l = -\infty$};

 	\node [blockzm2,  right=70pt of scaleinf ] (node1x) {$r$};
 	\node [blockzm2, below=25pt of node1x ] (node1a) {$r$};
 	\node [blockzm2, below of = node1a ] (node1b) {$r$};
 	\node [blockzm2,  below of = node1b ] (node1c) {$r$};
 	\node [blockzm2,   below of = node1c ] (node1d) {$r$};
 	\node [blockzm2,  below of = node1d ] (node1e) {$r$};
 	\node [blockzm2,   below =25pt of node1e ] (node1f) {$r$};
 	
 	\draw[dashed, ->] (node1x) -> (node1a);
 	    \draw[->] (node1a) -> (node1b);
 	     \draw[->] (node1b) -> (node1c);
 	      \draw[->] (node1c) -> (node1d);
 	       \draw[->] (node1d) -> (node1e);
 	      \draw[dashed, ->] (node1e) -> (node1f);

 		\node [blockzm2, right=10pt of node1b ] (node2b) {$p_{4}$};
 	\node [blockzm2,  below of = node2b ] (node2c) {$p_{4}$};
 	\node [blockzm2,   below of = node2c ] (node2d) {$p_{4}$};
 	\node [blockzm2,  below of = node2d ] (node2e) {$p_{4}$};
 	\node [blockzm2,   below =25pt of node2e ] (node2f) {$p_{4}$};
 	
 	    \draw[->] (node1a) -> (node2b);
 	    
 	    \draw[->] (node2b) -> (node2c);
 	      \draw[->] (node2c) -> (node2d);
 	       \draw[->] (node2d) -> (node2e);
 	      \draw[dashed, ->] (node2e) -> (node2f);

 		\node [blockzm2, right = 10pt of node2c ] (node3c) {$p_{3}$};
 	\node [blockzm2,   below of = node3c ] (node3d) {$p_{3}$};
 	\node [blockzm2,  below of = node3d ] (node3e) {$p_{3}$};
 	\node [blockzm2,   below =25pt of node3e ] (node3f) {$p_{3}$};
 	
 	        \draw[->] (node2b) -> (node3c);
 	      \draw[->] (node3c) -> (node3d);
 	       \draw[->] (node3d) -> (node3e);
 	      \draw[dashed, ->] (node3e) -> (node3f);

 		\node [blockzm2,   left = 10 pt of  node1d ] (node4d) {$p_{2}$};
 	\node [blockzm2,  below of = node4d ] (node4e) {$p_{2}$};
 	\node [blockzm2,   below =25pt of node4e ] (node4f) {$p_{2}$};

            \draw[->] (node1c)   -> (node4d);
 	       \draw[->] (node4d) -> (node4e);
 	      \draw[dashed, ->] (node4e) -> (node4f);

 	\node [blockzm2,  left = 10 pt of node4e ] (node5e) {$p_{1}$};
 	\node [blockzm2,   below =25pt of node5e ] (node5f) {$p_{1}$};
 	
 	\draw[->] (node4d) -> (node5e);
 	 \draw[dashed, ->] (node5e) -> (node5f);
 	
 	\node[rectangle, draw, fill=white, 
    text width=1.0em, text centered, rounded corners, minimum height=50pt, minimum width = 1.0em, right = 220pt of scalemidinf] (ynode1a) {$r$};
    
    	\node[rectangle, draw, fill=white, 
    text width=1.0em, text centered, rounded corners, minimum height=50pt, minimum width = 1.0em, below = 10pt of ynode1a] (ynode1b) {$r$};
    
    	\node[rectangle, draw, fill=white, 
    text width=1.0em, text centered, rounded corners, minimum height=80pt, minimum width = 1.0em, below = 10pt of ynode1b] (ynode1c) {$r$};

    \node[blockzm2, above right = -18 pt and 10 pt of ynode1b] (ynode2a) {$p_4$};
    
    	\node[rectangle, draw, fill=white, 
    text width=1.0em, text centered, rounded corners, minimum height=110pt, minimum width = 1.0em, below = 10pt of ynode2a] (ynode2b) {$p_4$};
    
    \node[rectangle, draw, fill=white, 
    text width=1.0em, text centered, rounded corners, minimum height=110pt, minimum width = 1.0em, right = 10pt of ynode2b] (ynode3a) {$p_3$};

        \node[blockzm2, above left = -15 pt and 10 pt of ynode1c] (ynode4a) {$p_2$};
    
    	\node[rectangle, draw, fill=white, 
    text width=1.0em, text centered, rounded corners, minimum height=55pt, minimum width = 1.0em, below = 10pt of ynode4a] (ynode4b) {$p_2$};
    
    \node[rectangle, draw, fill=white, 
    text width=1.0em, text centered, rounded corners, minimum height=55pt, minimum width = 1.0em, left = 10pt of ynode4b] (ynode5a) {$p_1$};

    \draw[->] (ynode1a) -> (ynode1b);
    \draw[->] (ynode1b) -> (ynode1c);
    \draw[->] (ynode1a) -> (ynode2a);
    \draw[->] (ynode2a) -| (ynode3a);
    \draw[->] (ynode2a) -> (ynode2b);
    
      \draw[->] (ynode1b) -> (ynode4a);
     \draw[->] (ynode4a) -| (ynode5a);
     \draw[->] (ynode4a) -> (ynode4b);
    

 	\node[blockzm2, right = 285pt of scale6] (xnode1) {$r$};
	\node[blockzm2, below right = 10 pt and 10 pt of xnode1] (xnode2) {$p_4$};
 	\node[blockzm2, below right = 70 pt and 10 pt of xnode1] (xnode4) {$p_2$};
 	\node[blockzm2, below right = 10 pt and 8 pt of xnode4] (xnode5) {$p_1$};
 	\node[blockzm2, below right = 12 pt and 8 pt of xnode2] (xnode3) {$p_3$};
 	
 	\draw[->]	(xnode1) -> (xnode2);
 	\draw[->]	(xnode2) -> (xnode3);
 	\draw[->]	(xnode1) -> (xnode4);
 	\draw[->]	(xnode4) -> (xnode5);

\end{tikzpicture}
   \caption{A comparison of past trees and a new compressed cover tree in Example \ref{exa:implicitexplicitexample}. \textbf{Left:} an implicit cover tree contains infinite repetitions of points. \textbf{Middle:} an explicit cover tree also with repeated points. \textbf{Right:} a compressed cover tree from Definition \ref{dfn:cover_tree_compressed} includes every point exactly once.  }
    \label{fig:tripleexample}
\end{figure}
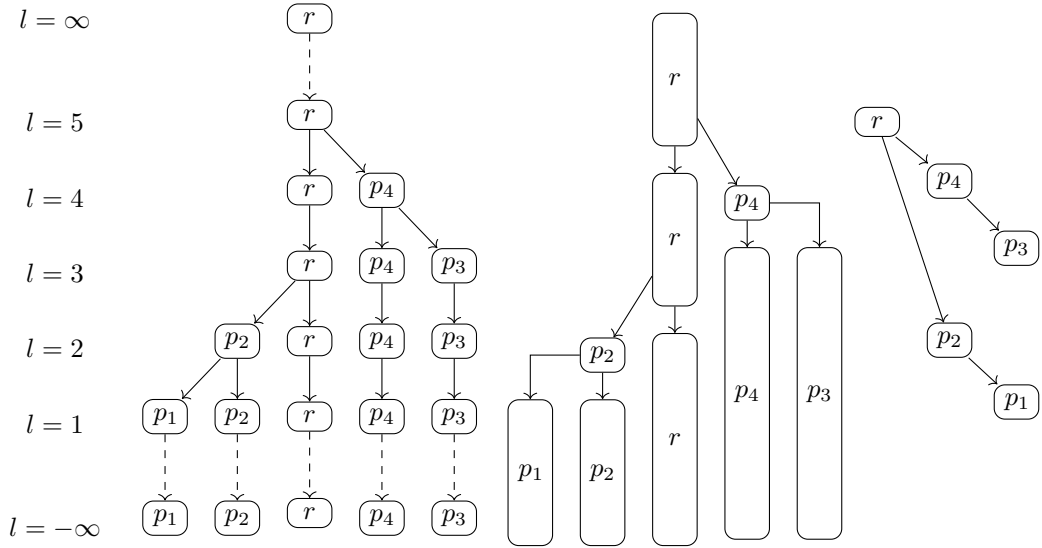

\begin{exa}[a short train line tree]
\label{exa:implicitexplicitexample}
Let $G$ be an unoriented metric graph consisting of two vertices $r,q$ connected by three different edges $e,h,g$ of lengths $|e| = 2^6$ , $|h| = 2^{3}$ , $|g| = 1$. 
Let $p_{4}$ be the middle point of the edge $e$. 
Let $p_{3}$ be the middle point of the subedge $(p_4 , q)$. 
Let $p_{2}$ be the middle point of the edge $h$.
Let $p_{1}$ be the middle point of subedge $(p_{2}, q)$. 
Set $R = \{p_1, p_2,p_3,p_4,r\}$. 
To construct a compressed cover tree $\T(R)$, set the level $l(p_i) = i$, choose the root $r$ to be the parent of both $p_2$ and $p_4$, $p_4$ to be the parent of $p_{3}$ and $p_{2}$ to be the parent of $p_{1}$. Then $\T(R)$ satisfies Definition~\ref{dfn:cover_tree_compressed}, see  three trees in Fig.~\ref{fig:tripleexample}.
\bs
\end{exa}

In any metric space $X$, let $\bar B(p,t)\subseteq X$ be the closed ball with a center $p$ and a radius $t$.
If this metric space is finite, $|\bar B(p,t)|$ denotes the number of points in $\bar B(p,t)$. 

\begin{dfn}[expansion constants $c$ \cite{beygelzimer2006cover} and $c_m$ {\cite[Definition~2.4]{elkin2021new}}]
\label{dfn:expansion_constant}
Let $R$ be a finite subset of a space $X$ with a metric $d$. 
The \emph{expansion constant} $c(R)$ is the smallest real number $c(R)\geq 2$ such that $|\bar{B}(p,2t)|\leq c(R) \cdot |\bar{B}(p,t)|$ for any $p\in R$ and radius $t\geq 0$.  
The \emph{minimized expansion constant} $c_m(R) = \inf\limits_{R\subseteq A\subseteq X}c(A)$ is minimized over all finite sets $A$ that cover $R$.  
\bs
\end{dfn}

\begin{restatable}[properties of $c_m$ {\cite[Lemma~2.5]{elkin2021new}}]{lem}{lemexpansionconstantproperty}
\label{lem:expansion_constant_property}
For any finite sets $R\subseteq U$ in a metric space, we have $c_m(R) \leq c_m(U)$, $c_m(R) \leq c(R)$.
\end{restatable}

\cite[Example~2.6]{elkin2021new} showed that the expansion constant $c(R)$ can be $O(|R|)$. Example~\ref{exa:minimized_normal_expansion_constant} shows that minimized expansion constant $c_m(R)$ can be much smaller than $c(R)$.

\begin{exa}[expansion constants]
\label{exa:minimized_normal_expansion_constant}
Let $(\R, d)$ be the Euclidean line. 
For any integer $n>10$, set $R = \{4^{i} \mid i \in [1,n]\}$ and let $Q = \{i \mid i \in [1,4^n]\}$. 
If $0<\epsilon < 10^{-9}$, then
 $\bar{B}(4^n, 2 \cdot 4^{n-1} - \epsilon) = \{2^n\} $ and $\bar{B}(2^n, 2( 2 \cdot 4^{n-1} - \epsilon)) = R$. 
 Therefore $c(R) = n$.
On the other hand, for any $q \in Q$ and any $t \in \R$, we have the balls $\bar{B}(q,t) = \mathbb{Z} \cap [q - t, q + t]$ and 
 $\bar{B}(q,t) = \mathbb{Z} \cap [q - 2t, q + 2t]$, hence $c(Q) \leq 4$.  Lemma \ref{lem:expansion_constant_property} implies that $c_m(R) \leq c_m(Q) \leq c(Q) \leq 4$.
\bs 
\end{exa}

\begin{restatable}[a distance bound on descendants {\cite[Lemma~2.8]{elkin2021new}}]{lem}{lemcompressedcovertreedescendantbound}
\label{lem:compressed_cover_tree_descendant_bound}
Let $R$ be a finite set with a metric $d$. 
In a compressed cover tree $\T(R)$, let $q$ be any descendant of a node $p$. Let the node-to-root path of $q$ have a node $u \in \Child(p) \setminus \{p\}$. Then $d(p,q) \leq 2^{l(u) + 2} \leq 2^{l(p) + 1}$. 
\bs
\end{restatable}

Lemma \ref{lem:packing} uses the idea of \cite[Lemma~1]{curtin2015plug} to show that if $S$ is a $\delta$-sparse subset of a metric space $X$, then $S$ has at most $(c_m(S))^\mu$ points in the ball $\bar{B}(p,r)$, where $c_m(S)$ is the minimized expansion constant of $S$, while $\mu$ depends on $\delta,r$. 

\begin{figure}
    \centering
    \input   \begin{tikzpicture}[scale = 1.0]




\node [circle,fill=red, red, inner sep=1pt, label = {$p$}] at (0,0) {};
\draw (0,0) circle[radius = 2];

\node [circle,fill=black, black, inner sep=1pt] at (-1,-1) {};
\draw[line width = .5pt, dash pattern=on 1pt off 2pt] (-1,-1) circle[ radius = 0.5];

\node [circle,fill=black, black, inner sep=1pt] at (1,0.75) {};
\draw[line width = .5pt, dash pattern=on 1pt off 2pt] (1,0.75) circle[ radius = 0.5];

 \draw[dashed] (1,0.75) -- (1.5,0.75);
 \node at (1.25,0.6) {\tiny $\delta / 2$}; 

\node [circle,fill=black, black, inner sep=1pt] at (-0.1,-0.4) {};
\draw[line width = .5pt, dash pattern=on 1pt off 2pt] (-0.1,-0.4) circle[ radius = 0.5];

\node [circle,fill=black, black, inner sep=1pt] at (0.75,-1.25) {};
\draw[line width = .5pt, dash pattern=on 1pt off 2pt] (0.75,-1.25) circle[ radius = 0.5];

\node [circle,fill=black, black, inner sep=1pt] at (-0.2,1.2) {};
\draw[line width = .5pt, dash pattern=on 1pt off 2pt] (-0.2,1.2) circle[ radius = 0.5];

\node [circle,fill=black, black, inner sep=1pt] at (-1.1,0.65) {};
\draw[line width = .5pt, dash pattern=on 1pt off 2pt] (-1.1,0.65) circle[ radius = 0.5];

 \draw[dashed] (0,0) -- node[below] {$t$} node[above] {} ++ (2,0);


\node at (-1.5,-0.25) {$S$};

\end{tikzpicture} 
    \caption{This volume argument proves Lemma~\ref{lem:packing}. By using an expansion constant, we can find an upper bound for the number of smaller balls of radius $\frac{\delta}{2}$ that can fit inside a larger ball $\bar{B}(p, t)$. }
    \label{fig:packingLemma}
\end{figure}
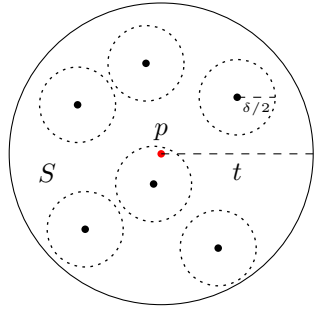

\begin{restatable}[packing {\cite[Lemma~2.9]{elkin2021new}}]{lem}{lempacking}
\label{lem:packing}
Let $S$ be a finite $\delta$-sparse set in a metric space $(X,d)$, so $d(a,b) > \delta$ for all $a,b \in S$. 
Then, for any point $ p \in X$ and any radius $t > \delta$, we have
$|\bar{B}(p, t)  \cap S | \leq (c_m(S))^{\mu}$, where $\mu = \lceil \log_2(\frac{4t}{\delta} + 1) \rceil $.
\bs
\end{restatable}

\cite[Section~1.1]{krauthgamer2004navigating} defined dim($X$) of a metric space $(X,d)$ as the minimum number $m$ such that every set $U \subseteq X$ can be covered by $2^{m}$ sets whose diameter is a half of the diameter of $U$. 
If $U$ is finite, Lemma \ref{lem:packing} for $\delta = \frac{r}{2}$ shows that
$\text{dim}(X) \leq \sup_{A \subseteq X}(c_m(A))^4 \leq \sup_{A \subseteq X}\inf_{A \subseteq B \subseteq X}(c(B))^4,$
where $A$ and $B$ are finite subsets of $X$. 
\medskip

Let $T(R)$ be an implicit cover tree \cite{beygelzimer2006cover} on a finite set $R$. 
\cite[Lemma~4.1]{beygelzimer2006cover} showed that the number of children of any node $p \in T(R)$ has the upper bound $(c(R))^4$. 
Lemma~\ref{lem:compressed_cover_tree_width_bound} generalizes \cite[Lemma~4.1]{beygelzimer2006cover} for a compressed cover tree.

\begin{restatable}[width bound {\cite[Lemma~2.10]{elkin2021new}}] {lem}{lemcompressedcovertreewidthbound}
\label{lem:compressed_cover_tree_width_bound}
Let $R$ be a finite subset of a metric space $(X,d)$. 
For any compressed cover tree $\T(R)$, any node $p$ has at most $(c_m(R))^4$ children at every level $i$, where $c_m(R)$ is the minimized expansion constant of the set $R$.
\bs
\end{restatable}

The following concept of a height was introduced and motivated in \cite[Definition~2.11]{elkin2021new}.

\begin{dfn}[the height of a compressed cover tree, {\cite[Definition~2.11]{elkin2021new}}]
\label{dfn:depth}
For a compressed cover tree $\T(R)$ on a finite set $R$,
the \emph{height set} is $H(\T(R))=\{l_{\max},l_{\min}\}\cup \{ i \mid C_{i-1} \setminus C_{i} \neq \emptyset\}$, whose size $|H(\T(R))|$ is called the \emph{height} of $\T(R)$.
\bs
\end{dfn}

By condition~(\ref{dfn:cover_tree_compressed}b), the height $|H(\T(R))|$ counts the number of levels $i$ whose cover sets $C_i$ include new points that were absent on higher levels.
Since any point can appear alone at its own level, $|H(\T)|\leq|R|$ is the worst case upper bound of the height.
The following parameters help prove an upper bound for the height $|H(\T(R))|$ in Lemma~\ref{lem:depth_bound}. 

\begin{dfn}[diameter and aspect ratio]
\label{dfn:radius+d_min}
For any finite set $R$ with a metric $d$, the \emph{diameter} is $\rad(R) = \max\limits_{p \in R}\max\limits_{q \in R}d(p,q)$.
The \emph{aspect ratio} \cite{krauthgamer2004navigating} is $\Delta(R) = \dfrac{\rad(R)}{d_{\min}(R)}$.
\bs
\end{dfn}

\begin{restatable}[an upper bound for the height $|H(\T(R))|$, see {\cite[Lemma~2.13]{elkin2021new}}]{lem}{lemdepthbound}
\label{lem:depth_bound}
Any finite set $R$ has the upper bound $|H(\T(R))|\leq 1+\log_2(\Delta(R))$.
\bs
\end{restatable}

If the aspect ratio $\Delta(R) = O(\text{Poly}(|R|))$ polynomially depends on the size $|R|$, then $|H(\T(R))| \leq O(\log(|R|))$.
All auxiliary lemmas are proved in appendix~\ref{sec:proofs}.

\section{Distinctive descendant sets and other tools for neighbor candidates}
\label{sec:distinctive_descendant_set}

This section we introduce a few more concepts and results from \cite{elkin2021new} for the new theorems in sections~\ref{sec:generaldualtreetheory} and~\ref{sec:CorrectDualTreeKNN}.
The first concept is a distinctive descendant set at a level $i$ of a node $p$ in a compressed cover tree.
Briefly, it is the set of descendants of a copy of $p$ at the level $i$ in the implicit cover tree on $R$.
Another important concept is a $\lambda$-point in Definition \ref{dfn:lambda-point}, which will be used in Algorithm~\ref{alg:UpdateCandidatesKNN} as an approximation for a $k$-nearest neighbor. 
Lemma~\ref{lem:beta_point} about a $\beta$-point will play a major role in the proof of main Theorem \ref{thm:paired_tree_knn_time}.

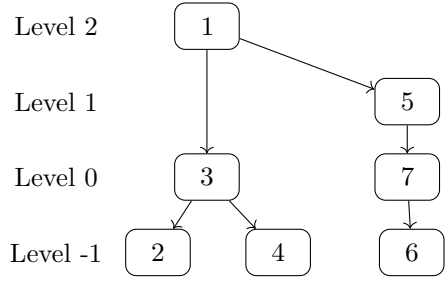
\begin{figure}
\centering
 \begin{tikzpicture}[align=center, node distance = 1.0cm, scale = 0.45]

	\node (scalem2text) {Level $2$};
	\node[below of =scalem2text] (scalem1text) {Level 1};
	\node[below of =scalem1text] (scalem0text) {Level 0};
	\node[below of =scalem0text] (scalemm1text) {Level -1};

	\node [blockz,  right=25pt of scalem2text ] (node1) {1};
	\node [blockz,  right=100pt of scalem1text] (node5) {5};
	\node [blockz,  right=25pt of scalem0text] (node3) {3};
	\node [blockz,  right=5pt of scalemm1text] (node2) {2};
	\node [blockz,  right=50pt of scalemm1text] (node4) {4};
	
	\node [blockz,  right=100pt of scalem0text] (node7) {7};
	\node [blockz,  right=100pt of scalemm1text] (node8) {6};

	  \draw[->] (node1) -> (node5);
	  \draw[->] (node1) -> (node3);
	  \draw[->] (node3) -> (node2);
	  \draw[->] (node3) -> (node4);
	  \draw[->] (node5) -> (node7);
	  \draw[->] (node7) -> (node8);

\end{tikzpicture} 
\caption{A compressed cover tree $\T(R)$ on the set $R = \{1,2,3,4,5,6,7\}$. Let $\Sd_i(p, \T(R))$ be a distinctive descendant set from Definition \ref{dfn:distinctive_descendant_set}. 
Then $V_2(1) = \emptyset, V_{1}(1) = \{5\}$, $V_{0}(1) = \{3,5,7\}$, and $\Sd_2(1, \T(R)) = \{1, 2,3,4,5,6,7\}$, $\Sd_1(1, \T(R)) = \{1,2,3,4\} $, $\Sd_{0}(1, \T(R)) = \{1\} $.}
\label{fig:uniqueDescendant}
\end{figure}

\begin{dfn}[Distinctive descendant sets]
\label{dfn:distinctive_descendant_set}
Let $R\subseteq X$ be a finite reference set with a cover tree $\T(R)$. 
For any node $p \in \T(R)$ in a compressed cove tree on a finite set $R$ and $i \leq l(p) - 1$, set
$V_{i}(p) =  \{u \in \Desc(p) \mid i \leq l(u)\leq l(p) - 1\}.$
If $i \geq l(p)$, then set $V_i(p) = \emptyset$. 
For any level $i \leq l(p) $, the \emph{distinctive descendant set} is
$\Sd_i(p, \T(R)) =   \Desc(p) \setminus \bigcup\limits_{u \in V_{i}(p)} \Desc(u)$ and has the size $|\Sd_i(p, \T(R)) |$. 
\bs
\end{dfn}

\begin{lem}[Distinctive descendant set inclusion property]
\label{lem:distinctive_descendant_inclusion}
In conditions of Definition \ref{dfn:distinctive_descendant_set} let $p \in R$ and let $i,j$ be integers satisfying
$l_{\min}(\T(R)) \leq i \leq j \leq l(p) - 1$.
Then $\Sd_i(p, \T(R)) \subseteq \Sd_j(p, \T(R))$.
\end{lem}




\begin{dfn}[$\Child(p,i)$ and $\nxt(p,i,\T(R))$ for a compressed cover tree]
 \label{dfn:implementation_compressed_cover_tree}
In a compressed cover tree $\T(R)$ on a set $R$, for any level $i$ and a node $p \in R$, set $\Child(p,i) = \{ a \in \Child(p) \mid l(a) = i \}$. 
Let $\nxt(p,i,\T(R))$ be the maximal level $j$ satisfying $j < i$ and $\Child(p,i) \neq \emptyset$. 
For every node $p$, we store its set of children in a linked hash map so that
 \begin{enumerate}[label=(\arabic*)]
     \item any key $i$ gives access to $\Child(p,i)$,
     \item every $\Child(p,i)$ has access to $\Child(p,\nxt(p,i, \T(R)))$,
     \item we can directly access $\max \{j \mid \Child(p,j) \neq \emptyset\}$.
     \bs
 \end{enumerate}
 \end{dfn}
 
 Essential levels of a node $p \in \T(R)$ has 1-1 correspondence to the set consisting of all nodes containing $p$ in the explicit representation of cover tree in \cite{beygelzimer2006cover}, see Figure \ref{fig:tripleexample} middle. 

\begin{dfn}[Essential levels of a node]
\label{dfn:essential_levels_node}
Let $R\subseteq X$ be a finite reference set with a cover tree $\T(R)$. Let $q \in \T(R)$ be a node. Let $(t_i)$ for 
$i \in \{0,1,...,n\}$ be a sequence of $H(\T(R))$ in such a way that $t_0 = l(q)$, $t_n = l_{\min}(\T(R))$ and for all $i$ we have $t_{i+1} = \nxt(q, t_{i}, \T(R))$. Define the set of essential indices $\Es(q,\T(R)) = \{ t_{i} + 1 \mid i \in \{0,...,n\} \}$.
\bs
\end{dfn}

\begin{lem}[Number of essential levels]
\label{lem:number_of_explicit_levels}
Let $R\subseteq X$ be a finite reference set with a cover tree $\T(R)$. Then 
$\sum_{p \in R}|\mathcal{E}(p,\T(R))| \leq 2 \cdot |R|,$
where $\mathcal{E}(p,\T(R))$ appears in Definition \ref{dfn:essential_levels_node}. \bs
\end{lem}
\begin{proof}
Let us prove this claim by induction on size $|R|$. In basecase $R = \{r\}$ and therefore $|\mathcal{E}(r,\T(R))| = 1$.
Assume now that the claim holds for any tree $\T(R)$, where $|R| = m$ and let us prove that if we add any node $v \in X \setminus R$ to tree $\T(R)$, then $\sum_{p \in R}|\mathcal{E}(p,\T(R \cup \{v\}))| \leq 2 \cdot |R| + 2$. Assume that we have added $u$ to $\T(R)$, in such a way that $v$ is its new parent. Then $|\mathcal{E}(p,\T(R \cup \{v\}))| = |\mathcal{E}(p,\T(R))| + 1 $ and $|\mathcal{E}(v,\T(R \cup \{v\}))| = 1$. We have:
$$ \sum_{p \in R \cup \{u\}}|\mathcal{E}(p,\T(R))| = \sum_{p \in R}|\mathcal{E}(p,\T(R))| + 1 + |\mathcal{E}(v,\T(R \cup \{v\}))| \leq 2\cdot |R| + 2 \leq 2(|R \cup \{v\}|) $$
which completes the induction step. 
\end{proof}

\begin{algorithm}
\caption{This algorithm returns sizes of distinctive descendant set $\Sd_i(p, \T(R))$ for all essential levels $i \in \Es(p,\T(R))$}
\label{alg:cover_tree_distinctive_descendants}
\begin{algorithmic}[1]
\STATE \textbf{Function} : CountDistinctiveDescendants(Node $p$, a level $i$ of $\T(R)$)
\STATE \textbf{Output} : Number of distinctive descendants of node $p$. 

\IF{$i > l_{\min}(\T(Q))$}
\FOR {$q \in \text{Children}(p)$ having $l(p) = i-1$ or $q = p$}
\STATE Set $s = 0$
\STATE $j \leftarrow 1 + \nxt(q, i-1,\T(R))$
\STATE $s \leftarrow s + $CountDistinctiveDescendants($q$, $j$)
\ENDFOR
\ELSE
\STATE Set $s = 1$
\ENDIF
\STATE Set $|\Sd_i(p)| = s$ and \textbf{return} s
\end{algorithmic}
\end{algorithm}

\begin{restatable}[{\cite[Lemma~5.3]{elkin2021new}}]{lem}{lemdistinctivedescendantsprecompute}
\label{lem:distinctive_descendants_precompute}
Let $R$ be a finite subset of a metric space.
Let $\T(R)$ be a compressed cover tree on $R$. 
Then, Algorithm~\ref{alg:cover_tree_distinctive_descendants} computes the sizes $|\Sd_{i}(p, \T(R))|$ for all $p \in R$ and essential levels $i \in \Es(p,\T(R))$ in time $O(|R|)$.
\bs
\end{restatable}

Let $i$ be an arbitrary level. Set $j =  \nxt(p,i, \T(R)) $.   
Definition~\ref{dfn:implementation_compressed_cover_tree} of $\nxt$ implies that $l(q) \notin [j+1,i-1]$ for all $q \in \Desc(p)$. Therefore we have $V_i(p) = V_j(p)$ and $\Sd_i(p, \T(R)) = \Sd_j(p, \T(R))$.  It follows that $|\Sd_i(p, \T(R))|$ can change only for the indices $i \in H(\T(R))$.  Lemma~\ref{lem:distinctive_descendants_precompute} shows that all the essential distinctive descendants sets of compressed cover tree $\T(R)$ can be precomputed in $O((c_m(R))^4\cdot|R| \cdot |H(\T(R))|)$ time.  
\medskip

Recall the neighborhood $N(q;r) = \{p \in C \mid d(q,p) \leq d(q,r)\}$ from Definition~\ref{dfn:kNearestNeighbor}.

\begin{dfn}[$\la$-point {\cite[Definition~5.6]{elkin2021new}}]
\label{dfn:lambda-point}
Fix a query point $q$ in a metric space $(X,d)$ and fix any level $i \in \Z$. 
Let $\T(R)$ be its compressed cover tree on a finite reference set $R \subseteq X$. 
Let $C$ be a subset of a cover set $C_i$ from Definition~\ref{dfn:cover_tree_compressed} satisfying $\sum_{p \in C}|\Sd_i(p, \T(R))| \geq k$, where $\Sd_i(p, \T(R))$ is the distinctive descendant set from Definition \ref{dfn:distinctive_descendant_set}.
For any $k\geq 1$, define $\la_k(q,C)$ as a point $\la\in C$ that minimizes $d(q,\la)$ subject to $\sum\limits_{p \in N(q;\la)}|\Sd_i(p, \T(R)) |\geq k$. 
\bs
\end{dfn}

\begin{algorithm}
\caption{Computing a $\lambda$-point of Definition \ref{dfn:lambda-point} in line 3 of Algorithm \ref{alg:UpdateCandidatesKNN},
 {\cite[Lemma~5.6]{elkin2021new}}}
\label{alg:lambda}
\begin{algorithmic}[1]
\STATE \textbf{Input:} A point $q \in X$, a subset $C$ of level set $C_i$ of a compressed cover tree $\T(R)$, an integer $k \in \Z$
\STATE Initialize an empty max-binary heap $B$ and an empty array $D$.
\FOR{$p \in C$}
\STATE add $p$ to $B$ with priority $d(q,p)$
\IF{$|H| \geq k$} 
\STATE remove the point with a maximal value from $B$
\ENDIF
\ENDFOR
\STATE Transfer points from the binary heap $B$ to the array $D$ in reverse order. 
\STATE Find the smallest index $j$ such that $\sum^{j}_{t = 0}\Sd_i(D[t] , \T(R)) \geq k$.
\STATE \textbf{return} $\lambda = D[j]$. 
\end{algorithmic}
\end{algorithm}

\begin{restatable}[time complexity of a $\lambda$-point {\cite[Lemma~5.7]{elkin2021new}}]{lem}{lemtimelambdapoint}
\label{lem:time_lambdapoint}
In the notations of Definition~\ref{dfn:lambda-point}, the time complexity of Algorithm \ref{alg:lambda} is $O(|C|\log(k)\max_i |R_i|)$.  \bs
\end{restatable}

\begin{restatable}[separation lemma {\cite[Lemma~5.8]{elkin2021new}}] {lem}{lemseparation}
\label{lem:separation}
In the conditions of Definition~\ref{dfn:distinctive_descendant_set}, let $p\neq q$ be nodes of $\T(R)$ with $l(p) \geq i$, $l(q) \geq i$. Then $\Sd_i(p , \T(R)) \cap \Sd_{i}(q, \T(R)) = \emptyset$.  \bs
\end{restatable}

\begin{restatable}[sum lemma {\cite[Lemma~5.9]{elkin2021new}}] {lem}{lemsum}
\label{lem:sum}
In the notations of Definition~\ref{dfn:lambda-point} for any subset $V \subseteq C$ of a set $C$, we have $|\bigcup\limits_{p \in V}\Sd_i(p,\T(R))| = \sum\limits_{p \in V} |\Sd_i(p, \T(R))|.$
\bs
\end{restatable}

By Lemma~\ref{lem:sum} we can assume in Definition~\ref{dfn:lambda-point} that $|\bigcup_{p \in C}\Sd_i(p, \T(R)) | \geq k$.

\begin{restatable}[{\cite[Lemma~5.11]{elkin2021new}}]{lem}{lemdistinctivedescendantchildlevel}
\label{lem:distinctive_descendant_child_level}
In the notations of Definition~\ref{dfn:distinctive_descendant_set}, let $p \in \T(R)$ be any node. 
If $w \in \Sd_i(p,\T(R))$ then either $w = p$ or there exists $a \in \Child(p) \setminus \{p\}$ such that $l(a) < i$ and $w \in \Desc(a)$.\bs
\end{restatable}

\begin{lem}
\label{lem:distinctive_descendant_distance}
In the notations of Definition~\ref{dfn:distinctive_descendant_set}, let $p \in \T(R)$ be any node. 
If $w \in \Sd_i(p,\T(R))$ then $d(w,p) \leq 2^{i+1}$.\bs
\end{lem}
\begin{proof}
By Lemma \ref{lem:distinctive_descendant_child_level} either $w = \gamma$ or $w \in \Desc(a)$ for some $a \in \Child(\gamma) \setminus \{\gamma\}$ for which $l(a) < i$.
If $w = \gamma$, then trivially $d(\gamma, w) \leq 2^{i}$. Else $w$ is a descendant of $a$, which is a child of node $\gamma$ on level $i-1$ or below, therefore by Lemma \ref{lem:compressed_cover_tree_descendant_bound} we have  $d(\gamma, w) \leq 2^{i}$ anyway.
\end{proof}

 \begin{restatable}[{\cite[Lemma~5.12]{elkin2021new}}] {lem}{lemchildsetequivalence}
\label{lem:child_set_equivalence}
Let $R$ be a finite subset of a metric pace.
Let $\T(R)$ be a compressed cover tree on $R$. 
Let $R_i \subseteq C_i$, where $C_i$ is the $i$th cover set of $\T(R)$. Set $\C(R_i) = \{a \in \Child(p) \text{ for some }p \in R_i \mid l(a) \geq i-1 \} $.
Then 
$$\bigcup_{p \in \C(R_i)}\Sd_{i-1}(p, \T(R)) = \bigcup_{p \in R_i}\Sd_{i}(p, \T(R)).$$
\end{restatable}

 \begin{restatable}[$\be$-point {\cite[Lemma~5.13]{elkin2021new}}]{lem}{lembetapoint}
\label{lem:beta_point}
In the notations of Definition~\ref{dfn:lambda-point}, let $C\subseteq C_i$ so that $\cup_{p \in C}\Sd_i(p, \T(R))$ contains all $k$-nearest neighbors of $q$. 
Set $\la = \la_k(q,C)$. 
Then $R$ has a point $\beta$ among the first $k$ nearest neighbors of $q$ such that $d(q,\lambda) \leq  d(q,\beta) + 2^{i+1}$.\bs
\end{restatable}

\section{A parametrized complexity of a traversal algorithm for paired trees}
\label{sec:generaldualtreetheory}

To fix the problem occurring in Counterexample \ref{cexa:dualtreeproof}, Definition \ref{dfn:imbalance} will introduce below a new concept of the imbalance inspired by \cite[Definitions~2-4]{curtin2015plug}. 
Any balanced tree from Example \ref{exa:cover_tree_big} will have a near linear imbalance $I(\T(R), \T(R)) = O(|R|)$. 
\medskip

Main Theorem~\ref{thm:paired_cover_tree_traversal_bound} of this section will express the time complexity of neighbor search based on paired trees in terms of $|Q|$, the height $|H(\T(R))|$, the new imbalance $I(\T(Q), \T(R))$,  and time complexities of auxiliary functions UpdateCandidates() and FinalCandidates(). 

\begin{dfn}[the imbalance $I$ of paired compressed cover trees]
\label{dfn:imbalance}
Let $R,Q$ be finite subsets of an ambient space. 
Let $\T(Q),\T(R)$ be compressed cover trees on $Q,R$, respectively.
Fro any $i\in\Z$, set $H_i(\T(R)) = H(\T(R)) \cap [l_{\min}(\T(R)) , i]$, where the height set $H(\T(R))$ was introduced in Definition \ref{dfn:depth}. 
Then the \emph{imbalance} of the paired trees $(\T(Q), \T(R))$ is 
 $$I(\T(Q), \T(R)) = \sum_{q \in \T(Q) }|H_{l(q)}(\T(R))|.\eqno{(\ref{dfn:imbalance})}$$ 
\end{dfn}

Definition~\ref{dfn:imbalance} combines the parameter $\theta$ from \cite[Definition~4]{curtin2015plug} and the original imbalance $I_t$ from \cite[Definition~3]{curtin2015plug} into the single new imbalance $I(\T(Q), \T(R))$. 
\medskip

Lemma~\ref{lem:imbalanceproperties} below gives a rough upper bound for the imbalance of paired cover trees. 
This rough bound weakens the new time complexity from main Theorem~\ref{thm:paired_tree_knn_time} based on paired trees
to the previous time complexity from \cite[Theorem~6.4]{elkin2021new} based on a single tree. 

\begin{lem}[imbalance bounds]
\label{lem:imbalanceproperties}
Let $R,Q$ be finite subsets of some ambient metric space. 
Let $\T(Q),\T(R)$ be compressed cover trees on $Q$,$R$, respectively. Then
$$I(\T(Q),\T(R)) \leq |Q| \cdot |H(\T(R))|.\eqno{(\ref{lem:imbalanceproperties})}$$
 If $\T(Q) = \T(R)$, then $I(\T(R),\T(R)) \geq |H(\T(R))|$ and $I(\T(R),\T(R)) \geq |R|$.
\end{lem}
\begin{proof}
The bound in (\ref{lem:imbalanceproperties}) is obtained as follows:
$$I(\T(Q),\T(R)) = \sum_{q \in \T(Q) }|H_{l(q)}(\T(R))| \leq \sum_{q \in \T(Q) } |H(\T(R))| \leq |Q| \cdot |H(\T(R))|.$$
Let $r$ be the root of the tree $\T(R)$. 
The second bound is obtained as follows:
$$I(\T(R),\T(R)) = \sum_{q \in \T(Q) }|H_{l(q)}(\T(R))| \geq |H_{l(r) }(\T(R))| \geq |H(\T(R))|.$$
Finally, 
$I(\T(R),\T(R)) = \sum\limits_{q \in \T(Q) }|H_{l(q)}(\T(R))| \geq \sum\limits_{q \in \T(Q) } 1 \geq |Q| = |R|.$
\end{proof}

\begin{exa}[imbalance computation]
The sets $Q = R = \{1,2,3,4,5\}$ have compressed cover trees $\T(R) = \T(Q)$ in Figure \ref{fig:implicitcompressed}. 
Since $\T(R)$ has nodes on levels $-1,0,1,2$, the height set is $H(\T(R)) = \{-1,0,1,2\}$.
The levels of all nodes are $l(1) = 2$, 
$l(5) = 1$, $l(3) = 0$, $l(2) = -1$ and $l(4) = -1$. 
Then $H_{2}(\T(R)) = H(\T(R))$, $ H_{1}(\T(R)) = \{-1,0,1\}$,
$H_{0}(\T(R)) = \{-1,0\}$, $H_{-1}(\T(R)) = \{-1\}$. 
By Definition \ref{dfn:imbalance} the imbalance is
$I(\T(Q), \T(R))  =  |H_{2}(\T(R))| + |H_{1}(\T(R))| + |H_{0}(\T(R))| + 2\cdot|H_{-1}(\T(R))| =  4 + 3 + 2 + 2 \cdot 1  = 11$, so
 $|Q| \cdot |H(\T(R))| = 5 \cdot 4 = 20$. 
\end{exa}

\begin{lem}[imbalance of a balanced tree pair]
\label{lem:perfect_tree_imbalance}
Let $R$ be a finite set with a metric and size $|R| = \frac{t^{m+1} - 1}{t - 1}$ for some integers $m,t$. 
Let a compressed cover tree $\T(R)$ have a structure of a balanced tree so that $H(\T(R)) = \{l_0, ..., l_m\}$, $|H(\T(R))| = m+1$, and $l_i < l_{i+1}$ for all $i$. 
Then for $0 \leq i \leq m$, the tree $\T(R)$ has exactly $t^i$ nodes at the level $l_i$. 
The imbalance is 
$$I(\T(R),\T(R)) =  \left(1+\frac{1}{t-1}\right)|R| - \frac{ |H(\T(R)) | }{t - 1} = O(|R|).$$ 
\end{lem}
\begin{proof}
The computations below follow directly from Definition \ref{dfn:imbalance}:
$I(\T(R),\T(R)) = \sum\limits^m_{i = 0}  |H_{l_i}(\T(R))|  \cdot t^{i} 
=  \sum\limits^m_{i = 0}(m-i+1) \cdot t^{i} = \dfrac{t^{m+2} - mt + m - 2t + 1}{(t-1)^2}
= |R| + \dfrac{t^{m+1} -mt + m - t}{(t-1)^2} = |R| + \dfrac{|R|}{t-1} +  \dfrac{ - mt + m - t + 1}{(t-1)^2}  
=  (1+\dfrac{1}{t-1})|R| - \dfrac{ |H(\T(R)) | }{t - 1} = O(|R|).$
\end{proof}

Algorithm \ref{alg:StandardCoverTreeTraversal} adapts \cite[Algorithm~1]{curtin2015plug} to compressed cover trees by replacing the functions BaseCase() and Score() by the new ones. 
In Algorithm~\ref{alg:StandardCoverTreeTraversal} the function FindCandidates() replaces the function BaseCase() to compute a final output for a given query point when we reach the lowest level of $\T(R)$. 
The function UpdateCandidates() replaces the function Score() to move neighbor candidates from $R_{i}$ to $R_{i-1}$. 
In Section \ref{sec:CorrectDualTreeKNN} Algorithm~\ref{alg:StandardCoverTreeTraversal} will be updated to 
Algorithm~\ref{alg:paired_tree_knn} to solve Problem \ref{pro:knn} finding all $k$-nearest neighbors of $Q$ in $R$. 
\medskip

We usually call Algorithm~\ref{alg:StandardCoverTreeTraversal} with the arguments $(l_{\max}(\T(R)), l_{\max}(\T(Q)), q, \{r\}) $, where $q,r$ are root nodes of compressed cover trees on $Q,R$, respectively. 

\begin{algorithm}
\caption{We adapt {\cite[Algorithm~1]{curtin2015plug}} for compressed cover trees, see Complexity Theorem~\ref{thm:paired_cover_tree_traversal_bound}.
}
\label{alg:StandardCoverTreeTraversal}
\begin{algorithmic}[1]
\STATE Function \textbf{PairedTreeTraversal}( a level $i \in \Z$, a level $j \in \Z$, a query node $q$ of $\T(Q)$, a subset $R_{i}$ of the cover set $C_i$ of $\T(R)$)
\IF {$i = l_{\min}(\T(R))$}
\STATE FinalCandidates$(i,j,q,R_i)$
\ENDIF
\IF {$\max(l_{\min}(\T(R)),j)  < i$ } \label{line:adt:beginReferenceExpansion}
\STATE \COMMENT{Reference expansion lines \ref{line:adt:beginReferenceExpansion}-\ref{line:adt:endReferenceExpansion}}
\STATE Set $\C(R_i) = \{a \in \Child(p) \text{ for some }p \in R_i \mid l(a) \geq i-1 \}$ \\ \COMMENT{recall that $\Child(p)$ contains the node $p$ }
\label{line:adt:dfnC}
\STATE $R_{i-1} \leftarrow \text{UpdateCandidates}(i,j,q,\C(R_i)) $\label{line:adt:SetRiDefinition}
\STATE  Set $t \leftarrow 1 + \max_{ a \in R_{i-1}} \nxt(a,i-1,\T(R)) $
\STATE Set $R_t \leftarrow R_{i-1}$
\STATE \textbf{call} \textbf{PairedTreeTraversal}$(q,j,R_{t})$ \COMMENT{recall that $R_t$ is a subset of $C_t$} \label{line:adt:endReferenceExpansion}
\ELSE  \label{line:adt:beginQueryExpansion}
\STATE \COMMENT{query expansion lines \ref{line:adt:beginQueryExpansion}-\ref{line:adt:endQueryExpansion}}
\FOR {$q' \in \text{Children}(q)$ having $l(q') = j-1$ or $q' = q$}
\STATE $j' \leftarrow 1 + \nxt(q', j,\T(Q))$
\STATE \textbf{call} \textbf{PairedTreeTraversal}$(q',j',R_i)$
\ENDFOR
\ENDIF \label{line:adt:endQueryExpansion}
\end{algorithmic}
\end{algorithm}


\begin{lem}[reference expansions in Algorithm \ref{alg:StandardCoverTreeTraversal}]
\label{lem:NumberOfReferenceExpansions}
Let $Q,R$ be finite subsets of a metric space. 
Let $\T(Q)$ and $\T(R)$ be compressed cover trees on $Q$ and $R$, respectively. 
Let us call Algorithm \ref{alg:StandardCoverTreeTraversal} with the input consisting of root nodes of $\T(Q)$, $\T(R)$ and their levels.
The algorithm runs at most $I(\T(Q), \T(R)) + |H(\T(R))|$ the reference expansions (lines \ref{line:adt:beginReferenceExpansion}-\ref{line:adt:endReferenceExpansion}).
\end{lem}
\begin{proof}
Let $\xi$ be the total number of reference expansions. 
Let $q$ be any non-root query node and let $p$ be its parent. 
We note that any call for PairedTreeTraversal($i, j, q, R_{i}$) comes  from the query expansion (lines 10-15) in PairedTreeTraversal($i'$, $l(q), p , R_{i'}$), see Algorithm~\ref{alg:StandardCoverTreeTraversal}.
\medskip

Then $l(q) \geq i' \geq i$. 
Since $q$ is the root node of $\T(Q)$, we can visit all levels from $H(\T(R))$, so
$\xi \leq \sum\limits_{q \in Q \setminus \text{root}(\T(Q))} |H_{l(q)}(\T(R))| + |H(\T(R))| \leq I(\T(Q), \T(R)) + |H(\T(R))|$.
\end{proof}

Theorem~\ref{thm:paired_cover_tree_traversal_bound} adapts \cite[Theorem~1]{curtin2015plug} to compressed cover trees and has the following advantages. 
\cite[Theorem~1]{curtin2015plug} 
used the single size $N = \max{|Q|, |R|}$ to estimate the time complexity. 
We replace this maximum size $N$ by $|Q| + |H(\T(R))|$, which is usually smaller than $N$ for online search in a large reference set $R$, especially if $Q$ is a single point. 

\begin{thmm}[time complexity of paired trees traversal]
\label{thm:paired_cover_tree_traversal_bound}
Let $Q, R$ be finite subsets of a metric space. 
Let $\T(Q)$ and $\T(R)$ be compressed cover trees on the sets $Q$ and $R$, respectively.
Let $B,S$ be the time complexities of the functions FinalCandidates() and UpdateCandidates() operation, respectively.
Let $I(\T(Q),\T(R))$ be the imbalance from Definition \ref{dfn:imbalance}. Then Algorithm~\ref{alg:StandardCoverTreeTraversal} has time complexity
$O\Big((B + S)\cdot\big( |H(\T(R))| + |Q| + I(\T(Q), \T(R)) \big)\Big)$. 
\end{thmm}
\begin{proof}
First, split the algorithm into two parts: reference expansions (lines \ref{line:adt:beginReferenceExpansion} -
\ref{line:adt:endReferenceExpansion}) and query
expansions (lines \ref{line:adt:beginQueryExpansion}–\ref{line:adt:endQueryExpansion}). 
The time complexity of the algorithm is bounded by the
total number of reference expansions times the maximum time of a reference expansion plus the total number of all query expansions times maximal time of a query expansion. 
\medskip

The time of a reference expansion (lines \ref{line:adt:beginReferenceExpansion} -
\ref{line:adt:endReferenceExpansion}) dominated by the sum $O(B+S)$ of the time complexities of the functions FinalCandidates() and 
UpdateCandidates(). 
Query expansions (lines~\ref{line:adt:beginQueryExpansion}–\ref{line:adt:endQueryExpansion}) are launched for every combiantion $q \in Q$ and $j \in \Es(q,\T(Q))$ of Definition \ref{dfn:essential_levels_node}. By Lemma \ref{lem:number_of_explicit_levels} their total number is $2\cdot|Q|$
Hence the total time complexity of all query expansions (lines \ref{line:adt:beginReferenceExpansion} -\ref{line:adt:endReferenceExpansion}) is $O(|Q|)$.
By Lemma \ref{lem:NumberOfReferenceExpansions} the number of reference expansions is at most $|H(\T(R))|+I(\T(R), \T(Q))$.
The final time complexity is the sum
$O( |Q|)+$ \\ $O((B + S )(|H(\T(R))| + I(\T(Q), \T(R)) ))
=  O\big((B + S)(|Q| + |H(\T(R))| + I(\T(Q), \T(R))) \big).$
\end{proof}

Theorem \ref{thm:paired_cover_tree_traversal_bound} holds even if a tree $\T(Q)$ fails the separation condition (\ref{dfn:cover_tree_compressed}c). 
This fact allows us to find a more balanced tree $\T(Q)$ to minimize the imbalance $I(\T(Q), \T(R))$ and substantially decrease the time complexity of the paired tree traversal in Algorithm~\ref{alg:StandardCoverTreeTraversal}.

\begin{cor}[case $Q = R$]
\label{cor:generalRunTime}
Let $R$ have a compressed cover tree $\T(R)$.
Let $B,S$ be the time complexities of the functions FinalCandidates() and UpdateCandidates(), respectively.
Algorithm \ref{alg:StandardCoverTreeTraversal} for $Q=R$ has time complexity $O\Big((B+S) \cdot I(\T(R), \T(R))\Big) $.
\end{cor}
\begin{proof}
It follows from Theorem \ref{thm:paired_cover_tree_traversal_bound} and Lemma \ref{lem:imbalanceproperties} as $I(\T(R), \T(R)) \geq |R| \geq H(\T(R))$.
\end{proof}

\section{A new algorithm for all $k$-nearest neighbors based on paired trees}
\label{sec:CorrectDualTreeKNN}

In this section we assume that a query set $Q$ and a reference set $R$ are finite subsets of an ambient metric space $X$ with a metric $d$.
Algorithm~\ref{alg:paired_tree_knn} will use paired trees $\T(R)$ and $\T(Q)$ to solve Problem~\ref{pro:knn} for any $k\geq 1$.
Main Theorem~\ref{thm:paired_tree_knn_time} proves a near linear time complexity for Algorithm~\ref{alg:paired_tree_knn}, which runs Algorithm~\ref{alg:StandardCoverTreeTraversal} by using Algorithm~\ref{alg:FinalCandidatesKNN} for the function FinalCandidates() and Algorithm \ref{alg:UpdateCandidatesKNN} for the function UpdateCandidates(). 
\medskip

Both functions FinalCandidates() and UpdateCandidates() need current levels $i,j$ of compressed cover trees $\T(R),\T(Q)$, respectively.
Other input arguments are a query point $q$ and a subset $R_i$ of a reference set $R$. 
In general, the function FinalCandidates() can use a level $j$ of a tree $\T(Q)$.
However, our implementation of Algorithm~\ref{alg:FinalCandidatesKNN} does not need $j$. 
\medskip

During the tree traversal for any $q \in Q$, we maintain the array $N_q$ (of candidate neighbors) consisting of pairs $(a,t)$, where $a \in R$ and $t = d(q,a)$. 
To minimize the run time of line $2$ in Algorithm~\ref{alg:FinalCandidatesKNN}, we implement the array $N_q$ as a max-heap \cite[section~6.5]{Cormen1990}, which has time complexity $O(\log(|N_q|))$ for adding elements, finding and removing a maximal element. 

\begin{algorithm}
\caption{This function FinalCandidates() is used in Algorithm \ref{alg:StandardCoverTreeTraversal} to get final Algorithm~\ref{alg:paired_tree_knn}. }
\label{alg:FinalCandidatesKNN}
\begin{algorithmic}[1]
\STATE \textbf{Input} : a level $i$ of $\T(R)$, a level $j$ of $\T(Q)$, a query point $q$, a subset $R_i \subset R$
\STATE Select the $k$-nearest neighbors of $q$ in the set $R_i$ and them into the array $N_q$ 
\end{algorithmic}
\end{algorithm}

\begin{algorithm}
\caption{This function UpdateCandidates() is used in Algorithm \ref{alg:StandardCoverTreeTraversal} to get Algorithm~\ref{alg:paired_tree_knn}. }
\label{alg:UpdateCandidatesKNN}
\begin{algorithmic}[1]
\STATE \textbf{Input} : a level $i$ of $\T(R)$, a level $j$ of $\T(Q)$, a query point $q\in Q$, \\ a subset $\mathcal{C}(R_i)$ of $C_{i-1}$ of a compressed cover tree $\T(R)$
\STATE \textbf{Output} : a subset $R_{i-1}\subset C_{i-1}$ consisting of all nodes that were not pruned. 
\STATE Compute a point $\lambda = \lambda_k(q,\C(R_i))$ from Definition~\ref{dfn:lambda-point} by Algorithm \ref{alg:lambda}. \label{line:UpdateCandidatesKNN:lambda}
\STATE \textbf{return} $R_{i-1} = \{a \in \C(R_i) \mid d(q,a) \leq d(q, \lambda) + 2^{i+1} + 2^{j+2}\} $
\label{line:dual_tree_knn_UpdateCandidates:last}
\end{algorithmic}
\end{algorithm}

\begin{algorithm}
\caption{The new algorithm for all $k$-nearest neighbors by using paired trees on $Q,R$}
\label{alg:paired_tree_knn}
\begin{algorithmic}[1]
\STATE \textbf{Input:} a query tree $\T(Q)$, a reference tree $\T(R)$, a number of neighbors $k\geq 1$
\STATE \textbf{Output:} arrays $N_q$ of $k$ nearest neighbors for all points $q \in Q$
\STATE run Algorithm \ref{alg:StandardCoverTreeTraversal} for the arguments
$( l_{\max}(\T(R)), l_{\max}(\T(Q)), \text{root}(\T(Q)), \{\text{root}(\T(R)) \} )$ by using \\
- the function UpdateCandidates() is implemented in Algorithm~\ref{alg:UpdateCandidatesKNN} \\
- the function FinalCandidates() is implemented in Algorithm~\ref{alg:FinalCandidatesKNN}
\STATE \textbf{return} the array $N_q$ for all points $q \in Q$. 
\end{algorithmic}
\end{algorithm}

\begin{lem}[true $k$-nearest neighbors are in a candidate set for all levels $i$]
\label{lem:cover_tree_dual_knn_correct}
Let $Q,R$ be finite subsets of a space $X$ with a metric $d$.
Let let $k\leq |R|$ be an integer. 
Assume that we have already constructed cover trees $\T(Q)$ and $\T(R)$ on $Q$ and $R$, respectively.
Consider any recursion $(i,j,q,R_i)$ of line $3$ in Algorithm \ref{alg:paired_tree_knn}. 
Then the union $\bigcup_{p \in R_i}\Sd_i(p, \T(R))$ contains all $k$-nearest neighbors of any point $q' \in \Desc(q)$.
\end{lem}
\begin{proof}
Assuming the contrary, let a recursion $(i,j,q,R_i)$ fail the required claim for some $i,j,q$. 
Define the order on $\Z \times \Z$ so that $(a,b) \leq (c,d)$ if $b < d$ or if $b = d$ and $a \leq b$. 
\medskip

For every point $q$, let $(i,j)_q$ be the maximal element for which $(i,j,q,R_i)$ doesn't satisfy the claim. 
Set $q = \mathrm{argmax}_{q \in Q}(i,j)_q$ and $(i,j) = \max_{q \in Q}(i,j)_q$.
Fix an arbitrary descendant $q' \in \Desc(q)$. 
Since $j \leq l(q) - 1$, Lemma~\ref{lem:compressed_cover_tree_width_bound} implies that $d(q',q) \leq 2^{j+1}$.
Choose
$$\beta \in \bigcup_{p \in R_i}\Sd_i(p, \T(R)) \setminus \bigcup_{p \in R_{i-1}}\Sd_{i-1}(p, \T(R)).$$ 
By Lemma~\ref{lem:child_set_equivalence} we have 
\begin{ceqn}
\begin{equation}
\label{eqa:neighborsContained}
\bigcup_{p \in \C(R_i)}\Sd_{i-1}(p, \T(R)) = \bigcup_{p \in R_i}\Sd_{i}(p, \T(R)).
\end{equation}
\end{ceqn}

Let $\lambda$ be a point computed in line 3 of Algorithm~\ref{alg:UpdateCandidatesKNN}. 
Equation (\ref{eqa:neighborsContained}) implies that $|\bigcup_{p \in \C(R_i)}\Sd_{i-1}(p, \T(R))| \geq k$.
Then the point $\lambda$ is well-defined by Definition~\ref{dfn:lambda-point}. 
Since $$\beta \in \bigcup_{p \in \C(R_i)}\Sd_{i-1}(p, \T(R)),$$ there exists $\alpha \in \C(R_i)$ satisfying $\beta \in \Sd_{i-1}(\alpha, \T(R))$. 
The initial assumption means that $\alpha \notin R_{i-1}$. 
Due to line~\ref{line:dual_tree_knn_UpdateCandidates:last} of Algorithm~\ref{alg:UpdateCandidatesKNN}, we get
\begin{ceqn}
\begin{equation}
\label{eqa:neighborsContained2}
   d(\alpha, q) > d(p, \lambda) + 2^{i+1} + 2^{j+2}.
\end{equation}
\end{ceqn}
Let $w$ be an arbitrary point in the union $\bigcup_{p \in N(q;\la)}\Sd_{i-1}(p, \T(R))$. 
Then $w \in \Sd_{i-1}(\gamma, \T(R))$ for some $\gamma \in  N(q;\la)$. By Lemma \ref{lem:distinctive_descendant_child_level} either $w = \gamma$ or $w \in \Desc(a)$ for a point $a \in \Child(\gamma) \setminus \{\gamma\}$ at a level $l(a) < i$.
If $w = \gamma$, then trivially $d(\gamma, w) \leq 2^{i}$. Else $w$ is a descendant of $a$, which is a child of node $\gamma$ on level $i-1$ or below.
\medskip

Then by Lemma \ref{lem:compressed_cover_tree_descendant_bound} we have  $d(\gamma, w) \leq 2^{i}$ anyway. 
By Definition \ref{dfn:lambda-point}  since $\gamma \in  N(q;\la)$ we have $d(q,\gamma) \leq d(q,\lambda)$. 
We apply (\ref{eqa:neighborsContained2}) and the triangle inequality below:
\begin{ceqn}
\begin{equation}
\label{eqa:neighborsContained3}
d(q',w) \leq d(q', q) + d(q, p) + d(p,w) \leq 2^{j+1} + d(\lambda, p) + 2^{i}  < d(\alpha,q) - 2^{i} - 2^{j+1}.
\end{equation}
\end{ceqn}
On the other hand, $\beta$ is a descendant of $\alpha$, so we can estimate
\begin{ceqn}
\begin{equation}
\label{eqa:neighborsContained4}
d(\beta,q') \geq d(\alpha,q)  - d(\alpha,\beta) - d(q,q')  \geq d(\alpha,q) - 2^{i} - 2^{j+1}. 
\end{equation}
\end{ceqn}
Inequalities~(\ref{eqa:neighborsContained3}),~(\ref{eqa:neighborsContained4}) imply that $d(q',w) < d(q',\beta)$. 
The point $w$ was arbitrarily chosen in the union $\bigcup_{p \in N(q;\la)}\Sd_{i-1}(p, \T(R))$ containing at least $k$ points.
Hence the point $\beta$ cannot be among the first $k$ nearest neighbors of $q'$.
This final contradiction proves the lemma. 
\end{proof}

\begin{thmm}[correctness of Algorithm~\ref{alg:paired_tree_knn} for all $k$-nearest neighbors]
\label{thm:paired_tree_knn_correctness}
Given any compressed cover trees $\T(Q)$ and $\T(R)$ on a query set $Q$ and a reference set $R$, respectively, Algorithm~\ref{alg:paired_tree_knn} solves Problem \ref{pro:knn} by finding all $k$-nearest neighbors in $R$ for all points $q\in Q$.
\end{thmm}
\begin{proof}
The nodes of $\T(R)$ cannot have children below the level
 $l_{\min}$, so $\bigcup_{p \in R_{l_{\min}}}\Sd_i(p, \T(R)) = R_{l_{\min}}$. Then all the $k$-nearest neighbors of any $q\in Q$ are contained in the set $R_{l_{\min}}$ obtained in the last recursion of Algorithm~\ref{alg:StandardCoverTreeTraversal} involving $q$.
Lemma~\ref{lem:cover_tree_dual_knn_correct} finishes the proof. 
\end{proof}

\begin{lem}[time complexity of the function FinalCandidates()]
\label{lem:FinalCandidates_knn_time}
For $k \leq |R|$, Algorithm~\ref{alg:FinalCandidatesKNN} with the input arguments $(i,j,q,R_i)$ has the time complexity $O(\log(k) \cdot |R_{i}|)$.
\end{lem}
\begin{proof}
The time complexity of Algorithm~\ref{alg:FinalCandidatesKNN} is dominated by line $2$, where we place points of $R_i$ into the array $N_q$ one by one until we find $k$ neighbors. 
At every step we iterate over the remaining points in the subset $R_i$ by adding them to the array $N_q$ one by one and by removing the node having a maximum distance to $q$. 
Since adding a new element and removing a maximal element from $N_q$ takes $O(\log(|N_q|))$ time, the time complexity of the function FinalCandidates() is $O(  |R_i| \cdot \log(|N_q|) )=O( \log(k) \cdot |R_{i}|)$ due to $|N_q| \leq k+1$.
\end{proof}

\begin{lem}[complexity of the function UpdateCandidates()]
\label{lem:UpdateCandidates_knn_time}
For $k \leq |R|$, Algorithm~\ref{alg:UpdateCandidatesKNN} with the input  arguments $(i,j,q,R_i)$ has the time complexity $O( c_m(R)^4 \cdot \log(k) \cdot |R_{i}|)$.
\end{lem}
\begin{proof}
Lemma~\ref{lem:compressed_cover_tree_width_bound} says that $|\C(R_i)| \leq (c_m(R))^4 \cdot |R_{i}|$.  
 By Lemma~\ref{lem:time_lambdapoint} line $3$ of Algorithm~\ref{alg:UpdateCandidatesKNN} takes time 
 $O(\log(k) \cdot \C(R_i)) \leq O(c_m(R)^4 \cdot \log(k) \cdot |R_{i}|)$. Line $4$ takes time $c_m(R)^4 \cdot |R_{i}|$.
So the time complexity is dominated by line $3$ and the lemma holds. 
\end{proof}

Theorem~\ref{thm:paired_tree_knn_time} will substantially extend {\cite[Theorem~2]{curtin2015plug}} solving Problem~\ref{pro:knn} only for $k1$ to compressed cover trees to find all $k$-nearest neighbors for any $k \geq 1$. 
We compare the time complexity $c(R)^4c^5_{qr}(N + I_t(\mathcal{T}_q)) + \theta)$ from \cite[Theorem~2]{curtin2015plug} with Theorem~\ref{thm:paired_tree_knn_time} below.
\medskip

The sum of the past parameters $I_t(\mathcal{T}_q)$ and $\theta$ is replaced by the new imbalance $I(\T(Q), \T(R))$. 
The size $N = \max(|Q|, |R|)$ reduces to $O(|Q| + |H(\T(R))|)$. 
The expansion constant $c(R)$ is replaced by the minimized expansion constant $c_m(R)$ from Definition~\ref{dfn:expansion_constant}. 
The factor $c^5_{qr}$ becomes $\max\{ c^3_{qr}k , c_m(R)^7\} \log(k)$.
The appearance of $k$ is natural, because Problem~\ref{pro:knn} is solved for any $k\geq 1$. 
The difference in powers of $c^{5}_{qr}$ and $c_m(R)^7$ is explained by the need to include more nodes into the subset $R_{i-1}$ in line $4$ of Algorithm \ref{alg:UpdateCandidatesKNN}.

\begin{thmm}[Complexity for $k$-nearest neighbors based on paired trees]
\label{thm:paired_tree_knn_time}
In the notations of Theorem \ref{thm:paired_tree_knn_correctness}, let  $c_{qr} = \max_{q \in Q}c(R \cup \{q\})$. Then the time complexity of Algorithm~\ref{alg:paired_tree_knn} which finds $k$-nearest neighbors for all $q \in Q$ in set $R$ as stated in Problem \ref{pro:knn} is
\begin{ceqn}
\begin{equation}
\tag{\ref{thm:paired_tree_knn_time}}
\label{eqa:ComputationTimeKNN}
 O\Big(c_m(R)^4 \cdot \max\{ c_{qr}^3k ,(c_m(R))^7\}\cdot \log(k)\cdot \big( |H(\T(R))| + |Q| + I(\T(R),\T(Q)) \big) \Big ),
\end{equation}
\end{ceqn}
where the imbalance $I(\T(R),\T(Q))$ was introduced in Definition \ref{dfn:imbalance}.
\end{thmm}
\begin{proof}
By Theorem~\ref{thm:paired_cover_tree_traversal_bound} the complexity of Algorithm \ref{alg:StandardCoverTreeTraversal} is bounded by
\begin{equation}
\label{eqa:paired_tree_knn_time_B+S}
O\Big((B + S)\cdot(I(\T(Q), \T(R)) + |H(\T(R))| + |Q|)\Big).
\end{equation}
Therefore it suffices to bound the complexity $B$ of the function  FinalCandidates() and the complexity $S$ of the function UpdateCandidates(). 
By Lemma~\ref{lem:FinalCandidates_knn_time} and Lemma~\ref{lem:UpdateCandidates_knn_time} we have $$B + S =  O(\log(k) \cdot |R_{i}|) +  O( c_m(R)^4 \cdot \log(k) \cdot |R_{i}|)  \leq O(c_m(R)^4 \cdot \log(k) \cdot |R_{i}|).$$

Let us now bound the maximal size of the subset $R_{i}$. 
If $i = l_{\max}(\T(R))$ then $|R_i| = 1$. 
If $i < l_{\max}(\T(R))$, it suffices to estimate the size $|R_{i-1}|$.
Consider $R_{i-1} = \{a \in \C(R_i) \mid d(q, a) \leq d + 2^{j+2} + 2^{i+1}\},$ where $d = d(q,\lambda_k(q,\C(R_i)))$. 
Since $|R_i|$ can only change in the reference expansion lines \ref{line:adt:beginReferenceExpansion}-\ref{line:adt:endReferenceExpansion} of Algorithm \ref{alg:StandardCoverTreeTraversal}, we get $j + 1 \leq i$.  
Then $\C(R_{i}) \subseteq C_{i-1}$ gives
\begin{equation}
\label{eqa:QBound}
R_{i-1} = B(q, d+2^{i+1} + 2^{j+2}) \cap \C(R_i) \subseteq B(q, d+2^{i+2}) \cap C_{i-1}.
\end{equation}

\noindent
\textbf{Case} $d \geq 2^{i+2}$. 
By Lemma \ref{lem:cover_tree_dual_knn_correct} the set $\cup_{p \in R_i}\Sd_i(p, \T(R))$ contains all $k$-nearest neighbors of $q$.
By Lemma \ref{lem:child_set_equivalence} we have $$\cup_{p \in \C(R_i)}\Sd_{i-1}(p, \T(R)) = \cup_{p \in R_i}\Sd_{i}(p, \T(R)).$$ Therefore 
 $\cup_{p \in \C(R_i)}\Sd_{i-1}(p, \T(R))$ contains all $k$-nearest neighbors of $q$. 
\medskip
 
Lemma~\ref{lem:beta_point} applied for the index $i-1$ says that the reference set $R$ has a point $\beta$ among all $k$-nearest neighbors of $q$ such that $2^{i+2} \leq d \leq d(q,\beta) + 2^{i}$. 
Then $3\cdot2^{i} \leq d(q,\beta)$ and $$d + 2^{i+2} \leq d(q,\beta) + 2^{i} + 2^{i+2} \leq (1+\frac{5}{3})d(q,\beta) \leq 4d(q,\beta).$$
Combining the above estimate with inclusion~(\ref{eqa:QBound}) and Definition~\ref{dfn:expansion_constant}, we get
$$|R_{i-1}|  \leq |B(q,d(q,\beta) + 2^{i} + 2^{i+2})| \leq |B(q,4d(q,\beta))| \leq c_{qr}^3|B(q,\frac{d(q,\beta)}{2})|.$$
If $|B(q,\frac{d(q,\beta)}{2})| \geq k$, then $B(q,\frac{d(q,\beta)}{2})$ contains $k$ points $\alpha$ with $d(q,\alpha) \leq \frac{d(q,\beta)}{2}$. 
Then the point $\beta$ cannot be among the first $k$ nearest neighbors of $q$, which contradicts a choice of $\be$ above. 
Then $|B(q,\frac{d(q,\beta)}{2})| < k$, so $|R_{i-1}| \leq c_{qr}^3 \cdot k$. 
\medskip

\noindent
\textbf{Case} $d \leq 2^{i+2}$. 
Set $ t = 2^{i+3}$ . 
Inclusion~(\ref{eqa:QBound}) implies that
$$R_{i-1} \subseteq B(q,d + 2^{i+2}) \subseteq B(q,t) \cap C_{i-1}.$$
By separation condition~(\ref{dfn:cover_tree_compressed}c) for $\T(R)$, all points in $C_{i-1}$ are separated by $\delta = 2^{i-1}$. Then $\frac{4t}{\delta} + 1 = 2^{6} + 1 \leq 2^{7}$.
Lemma~\ref{lem:packing} implies that $|R_{i-1}| \leq |B(p,t) \cap C_{i-1}| \leq (c_m(R))^7$.
\medskip

The two cases above provide the upper bound $|R_{i-1}| \leq \max\{c_{qr}^3 k, (c_{m}(R))^7 \}$ for any $i$. 
\begin{equation*}
\text{Then }O\Big(B + S \Big) =O\Big( c_m(R)^4 \cdot \log(k) \cdot |R_{i}| \Big) = O\Big(c_m(R)^4 \cdot \max\{ c_{qr}^3k ,(c_m(R))^7\}\cdot \log(k)\Big)
\end{equation*}
is obtained from the time complexity in (\ref{eqa:paired_tree_knn_time_B+S}) after replacing $|R_{i-1}|$ by its upper bound.
\end{proof}

\begin{cor}[simplified time complexity for all $k$-nearest neighbors based on paired trees]
\label{cor:paired_tree_knn_time}
The time complexity of Algorithm~\ref{alg:paired_tree_knn} from  Theorem \ref{cor:paired_tree_knn_time} has the simpler upper bound
\begin{ceqn}
\begin{equation}
\tag{\ref{cor:paired_tree_knn_time}}
\label{eqa:cor:ComputationTimeKNN}
 O\Big(c^{11}_{qr} \cdot k\log(k)\cdot \big( |H(\T(R))| + |Q| + I(\T(R),\T(Q)) \big) \Big ),
\end{equation}
\end{ceqn}
see the height $|H(\T(R))|$ and the imbalance $I(\T(R),\T(Q))$ in Definitions~\ref{dfn:depth} and~\ref{dfn:imbalance}.
\end{cor}

\section{Conclusions, discussion and further applications of the new results}
\label{sec:Conclusions}

This paper was motivated by the remaining challenges in the neighbor search based on paired trees, see Problem~\ref{pro:knn} in section~\ref{sec:intro}. 
Section~\ref{sec:review} described key steps in the past proofs that require more justifications, see detailed Counterexamples~\ref{cexa:dualtreecode} and~\ref{cexa:dualtreeproof} in appendix~\ref{sec:challenges}.
\medskip

Main Theorem~\ref{thm:paired_tree_knn_time} has completely solved Problem~\ref{pro:knn} for any $k\geq 1$ in any metric space and also improved the recent time complexity based on a single tree in \cite[Theorem~6.5]{elkin2021new}.
\medskip

In comparison with the 2015 estimate in \cite[Theorem~2]{curtin2015plug}, Theorem~\ref{thm:paired_tree_knn_time} extended the time complexity to any number of neighbors $k\geq 1$.
Even for $k=1$, the complexity was improved due to the new height and imbalance.
These parameters capture the hardness of the query and reference sets for neighbor search.
If all parameters have constant bounds, the final complexity $O(k\log k\max\{|Q|,|R|\})$ is near linear with respect to the key variables.
\medskip

As a practical application, the time complexities in \cite[Theorem~6.5]{elkin2021new} and Theorem~\ref{thm:paired_tree_knn_time} helped prove a near linear time complexity for new invariants of crystals in \cite[Theorem~14]{widdowson2022average}.
These invariants are formed by distances from atoms to their nearest neighbors in an infinite periodic crystal.
A modest desktop over a couple of days completed 400M+ pairwise comparisons of the above neighbor-based invariants for all 660K+ periodic crystals in the world's largest Cambridge Structural Database.
The resulting five pairs of impossible duplicates are now investigated by the journals that published the underlying papers.
\medskip

One can extend the developed methods in several directions.
Similarly to \cite[Theorem~7.4]{elkin2021new}, we conjecture that the $(1+\epsilon)$-approximate version of Problem \ref{pro:knn} can be solved in time $O\Big(c_m(R)^{\log_2(O(1/\ep))} \log(k) \cdot ( I(\T(R), \T(Q)) + \log_2(\Delta) + |Q|) + k \Big)$, where $c_m(R)$ is the minimized expansion constant from Definition~\ref{dfn:expansion_constant}, $\Delta$ is the aspect ratio of $R$, see Definition~\ref{dfn:radius+d_min}. 
\medskip

\cite[Section~2.3.1]{beygelzimer2006coverExtend} hints at a possibility to improve \cite[Algorithm~4.3]{elkin2021new} for a compressed cover tree to reduce the complexity $O(c_m(R)^{O(1)} \cdot |R| \cdot |H(\T(R))|)$ to $O(c_m(R)^{O(1)}\cdot I( \T(R) , \T(R))$.
The time complexity for a Minimum Spanning Tree \cite{EMST} estimates the number of recursions with the same issues as in the proof of \cite[Theorem~5]{beygelzimer2006cover} discussed in section~\ref{sec:review}.
The challenging datasets from Counterexamples~\ref{cexa:dualtreecode} and~\ref{cexa:dualtreeproof} in appendix~\ref{sec:challenges}
 imply that the proof of \cite[Theorem~5.1]{EMST} can be improved by using the novel concept of the imbalance of paired trees. 
\medskip

This research was supported by the £3.5M EPSRC grant ‘Application-driven Topological Data Analysis’ (2018-2023, EP/R018472/1), the £10M Leverhulme Research Centre for Functional Materials Design and the last author’s Royal Academy of Engineering Fellowship ‘Data Science for Next Generation Engineering of Solid Crystalline Materials’ (IF2122/186).
\medskip

The authors are grateful to all reviewers for their valuable time and comments in advance.

\bibliographystyle{plainurl}
\bibliography{neighbor_search}


\appendix

\section{Challenging data for a nearest neighbor search based on paired trees}
\label{sec:challenges}

This section discusses $k$-nearest neighbor Algorithm~\ref{alg:cover_tree_k-nearest_dt_original} based on paired trees from \cite{ram2009linear}. 
\medskip

\begin{algorithm}
\caption{
Original \cite[Algorithm~1]{ram2009linear} is analyzed in Counterexamples~\ref{cexa:dualtreecode} and~\ref{cexa:dualtreeproof}.
}
\label{alg:cover_tree_k-nearest_dt_original}
\begin{algorithmic}[1]
\STATE \textbf{Function} FindAllNN(a node $q_j\in T(Q)$, a subset $R_i$ of a cover set $C_i$ of $T(R)$).
\IF {$i = -\infty$}
\STATE for each $q_j \in L(q_j)$ \textbf{return} $\text{argmin}_{r \in R_{-\infty}} d(q,r)$
\STATE \COMMENT{here $L(q_j)$ is the set of all descendants of the node $q_j$}
\ELSIF{$j < i$}
\STATE $\C(R_i) = \{\text{Children}(r) \mid r \in R_i\} $ \COMMENT{in original pseudo-code the notation is $R = \C(R_i)$}
\STATE $R_{i-1} = \{r \in R \mid d(q_j,r) \leq d(q_j, R) + 2^{i} + 2^{j+2} \}$
\STATE FindAllNN($q_{j-1}, R_i$) \COMMENT{ $q_{j-1}$ is the same point as $q_j$ on one level below}
\ELSE{}
\STATE for each $p_{j-1} \in \text{Children}(q_j)$ FindAllNN($p_{j-1},R_{i}$)
\ENDIF
\end{algorithmic}
\end{algorithm}

Counterexample \ref{cexa:dualtreecode} shows how Algorithm \ref{alg:cover_tree_k-nearest_dt_original} cannot be used for finding non-trivial nearest neighbors in the case $Q = R$. 
Counterexample \ref{cexa:dualtreeproof} provides a cover tree that should, according to the proof of \cite[Theorem~3.1]{ram2009linear}, run at most $O(\max\{|Q|,|R|\} \cdot \sqrt{\max\{|Q|,|R|\}}$ reference expansions (lines 5-8) during the whole execution of Algorithm~\ref{alg:cover_tree_k-nearest_dt_original}. 
\medskip

The step-by-step execution of Algorithm~\ref{alg:cover_tree_k-nearest_dt_original} will show that the number of reference expansions has a lower bound $O(\max\{|Q|,|R|\}^2)$. 
These issues are resolved in sections \ref{sec:generaldualtreetheory} and \ref{sec:CorrectDualTreeKNN} by rewriting the algorithm for compressed cover trees and by using the new parameters: the height of a compressed cover tree in Definition \ref{dfn:depth} and the imbalance in Definition~\ref{dfn:imbalance}. 
\medskip
 
Recall that \cite[End of Section~1]{ram2009linear} defined the all-nearest-neighbor problem as follows. 
"\textbf{All Nearest-neighbors:} For all queries $q \in Q$ find $r^{*}(q) \in R$ such that $r^{*}(q) = \mathrm{argmin}_{r \in R}d(q,r)"$. 
For $Q=R$, the last formula produces trivial self-neighbors.
\medskip
 
In original Algorithm~\ref{alg:cover_tree_k-nearest_dt_original}, the node $q_j$ has a level $j$, a reference subset $R_i\subset R$ is a subset of $C_i$ for an implicit cover tree $T(R)$. 
The algorithm is called for a pair $q_j, R_{i} = \{r\}$, where $q_j$ is the root of the query tree at the maximal level $j = l_{\max}(T(Q))$,
and $r$ is the root of the reference tree at the maximal level $i = l_{\max}(T(R))$. 
Split Algorithm~\ref{alg:cover_tree_k-nearest_dt_original} into these blocks:
\smallskip

\noindent
lines 2-4 : FinalCandidates,
\smallskip

\noindent
lines 5-9 : reference expansion,
\smallskip

\noindent
lines 9-11 : query expansion. 

\begin{exa}[tall imbalanced tree {\cite[Example~3.1]{elkin2021new}}]
\label{exa:tall_imbalanced_tree}
For any integer $m > 10$, let $G$ be a metric graph pictured in Figure \ref{fig:GraphConstructionOfExample} that has two vertices $r,q$ and $m+1$ edges $(e_i)$ for $i \in \{0, ..., m\}$, and the length of each edge $e_i$ is $|e_i| = 2^{m \cdot i +2}$ for $i \geq 1$.
Finally, set $|e_0| = 1$. 
\medskip

For every $i \in \{1, ..., m^2\}$ if $i$ is divisible by $m$ we set $p_{i}$ be the middle point of $e_{i / m}$ and for every other $i$ we define $p_i$ to be the middle point of segment $(p_{i+1}, q)$. 
\medskip

Let $d$ be the induced shortest path metric on the continuous graph $G$.
Then $d(q,r) = 1$, $d(r, p_{i}) = 2^{i+1} + 1$, $d(q,p_{i}) = 2^{i} $ and if $i > j $ and $\ceil{\frac{i}{m}} = \ceil{\frac{j}{m}}$ we have $ d(p_{j}, p_{i})  = \sum\limits_{t = j+1}^{i} 2^{t}$. 
We consider the reference set $R = \{r\} \cup \{p_{i} \mid i \in \{1,2,3,...,m^2\} \}$ with the metric $d$.
\medskip

Let us define a compressed cover tree $\T(R)$ by setting $r$ to be the root node and $l(p_{i}) = i$ for all $i$. 
If $i$ is divisible by $m$, we set  $r$ to be the parent of $p_{i}$. 
If $i$ is not divisible by $m$, we set $p_{i+1}$ to be the parent of $p_{i}$. 
For every $i$ divisible by $m$, the point $p_i$ is in the middle of edge $e_{i / m}$, hence $d(p_{i}, r) \leq 2^{i + 1}$.
\medskip
 
For every $i$ not divisible by $m$, by the definition $p_i$ is middle point of $(p_{i+1},q)$ and therefore we have $d(p_i, p_{i+1}) \leq 2^{i+1}$. Since for any point $p_i$ distance to its parent is at most $2^{i+1}$, it follows that $\T(R)$ satisfies the covering condition of Definition \ref{dfn:cover_tree_compressed}. 
\medskip

For any integer $t$, the cover set is $C_t = \{r\} \cup \{ p_i \mid i \geq t\}$. 
We will prove that $C_t$ satisfies the separation condition. 
Let $p_{i} \in C_t$.
If $i$ is divisible by $m$, then 
$d(r, p_i) = 2^{i+1} \geq 2^{t + 1} > 2^t.$ 
If $i$ is not divisible by $m$, then
$d(r, p_{i}) = d(r,q) + d(q,p_{i})  = 1 + 2^{i+1}  > 2^{t}$. 
\medskip

Then the root $r$ is separated from the other points by the distance $2^t$. 
Consider arbitrary points $p_{i}$ and $p_{j}$ with indices $i > j \geq t $ and $\ceil{\frac{i}{m}} = \ceil{\frac{j}{m}}$.
Then
$$d(p_{i}, p_{j}) = \sum^{i}_{s = j+1} 2^{s}  \geq 2^{j+1} \geq 2^{t+1} > 2^{t}.$$
On the other hand, if $i > j \geq t $ and $\ceil{\frac{i}{m}} \neq \ceil{\frac{j}{m}}$, then
$$d(p_{i} , p_{j}) = d(p_{i},q) + d(p_{j} ,q)  \geq  2^{i} + 2^{j} \geq 2^{j+1} \geq 2^{t+1} > 2^t.$$
For any $t$, we have shown that all pairwise combinations of points of $C_{t}$ satisfy the separation condition in Definition~\ref{dfn:cover_tree_compressed}.
Hence this condition holds for the whole tree $\T(R)$. 
\bs
\end{exa}

Recall that in \cite[Section~2]{beygelzimer2006cover} the explicit representation of cover tree was defined as 
"the explicit representation of the tree coalesces all nodes in which the only child is a self-child". Simplest way to interpret this is to consider cover sets $C_i$ and define $p \in C_i$ to be an explicit node, if $p$ has child at level $i-1$.  
By \cite[Lemma~4.3]{beygelzimer2006cover} the depth of any node $p$ is "defined as the number of explicit grandparent nodes on the path from the root to $p$ in the lowest level in which $p$ is explicit". 
The explicit depth of a node $p$ in any compressed tree $\T$ will be defined in Definition \ref{dfn:explicit_depth_for_compressed_cover_tree} using the simplest interpretation of the quotes above.

\begin{lem}[explicit depth for a compressed cover tree {\cite[Lemma~3.2]{elkin2021new}}]
\label{dfn:explicit_depth_for_compressed_cover_tree}
Let $R$ be a finite subset of a metric space with a metric $d$.
Let $\T(R)$ be a compressed cover tree on $R$. 
For any $p \in \T(R)$, let $s = (w_0, ... , w_m)$ be a node-to-root path of $p$. 
Then the explicit depth $D(p)$ of node $p$ belonging to compressed cover tree can be interpreted as the sum 
$$D(p) = \sum^{m-1}_{i = 0}| \{q \in \Child(w_{i+1}) \mid l(q) \in [l(w_i), l(w_{i+1}) - 1] \} |. $$
\bs
\end{lem}
\begin{proof}
Note that the node-to-root path of an implicit cover tree on $R$ has $l(w_{j+1}) - l(w_{j}) - 1$ extra copies of $w_{j+1}$ between every $w_{j}$ and $w_{j+1}$ for any index $j \in [0,m-1]$. Recall that a node is called explicit, if it has non-trivial children. 
Therefore there will be exactly $$| \{q \in \Child(w_{i+1}) \mid l(q) \in [l(w_i), l(w_{i+1}) - 1] \} | $$ 
explicit nodes between $w_{j}$ and $w_{j+1}$.
It remains to take the total sum.
\end{proof}

\begin{lem}[{\cite[Lemma~3.3]{elkin2021new}}]
\label{lem:tall_imbalanced_tree_explicit_depth}
Let $\T(R)$ be a compressed cover tree on the set $R$ from Example \ref{exa:tall_imbalanced_tree} for some $m \in \Z$. 
For any $p \in R$, the explicit depth $D(p)$ of Definition~\ref{dfn:explicit_depth_for_compressed_cover_tree} has the upper bound $2m+1$. 
\bs
\end{lem}
\begin{proof}
For any $p_i$, if $i$ is divisible by $m$, then $r$ is the parent of $p_i$. By definition, the explicit depth is $D(p_i) = |\{p \in \Child(r) \mid l(p) \in [l(p_i), m^2] \} |$. 
Since $r$ contains children on every level $j$, where $j$ is divisible by $m$, we have $D(p_i) =  m - \frac{i}{m} + 1$. 
\medskip

Let us now consider an index $i$ that is not divisible by $m$. 
Note that  $p_{j+1}$ is the parent of $p_j$ for all $j \in [i, m \cdot \ceil{i / m} - 1]$.  
Then the path consisting of all ancestors of $p_i$ from $p_i$ to the root node $r$ has the form $(p_{i}, p_{i+1}, ..., p_{m \cdot \ceil{i / m}} , r)$. It follows that
$$D(p_i) = \sum^{m \cdot \ceil{i / m} -1}_{j = i}| \{p \in \Child(p_j) \mid l(p) \in [l(p_{j}), l(p_{j+1}) - 1] \}| + D(p_{m \cdot \ceil{i / m}}).
 $$
Since $i \geq m \cdot (\ceil{i / m} -1) + 1 $ and $\ceil{i / m} \geq 1$, we get the required upper bound: $$D(p_i) =  (m \cdot  \ceil{i / m} - i ) + (m - \ceil{i / m} + 1) \leq m + (m+1) = 2m + 1.$$ 
\end{proof}

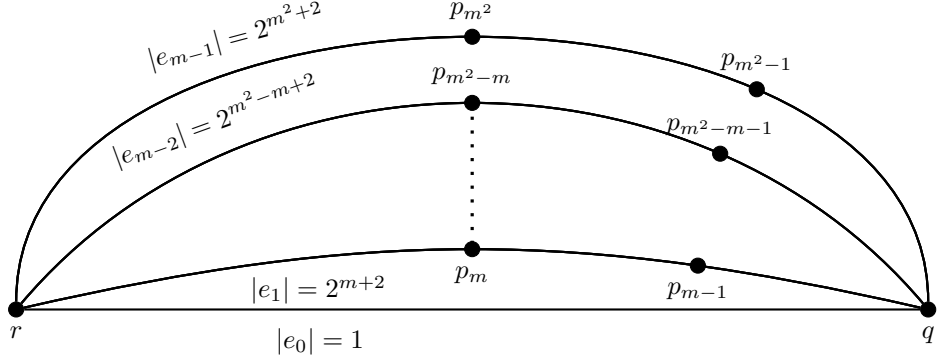
\begin{figure}
    \centering
    \tikzset{markpos/.style args={#1 at #2}{decoration={
  markings,
  mark=at position #2 with {\coordinate(#1);}},postaction={decorate}}}
  
\begin{tikzpicture}
\node[circle,fill=black,inner sep=2pt,draw,label = below:{$r$}] (a) at (180:6cm) {};
\node[circle,fill=black,inner sep=2pt,draw,label = below:{$q$}] (b) at (0:6cm) {};
\draw[thick] (a) edge[bend left=90, markpos=mymark1 at 0.5] (b);
\draw[thick] (a) edge[bend left=90, markpos=edgemark1 at 0.30] (b);
\draw[thick] (a) edge[bend left=90, markpos=anothermark1 at 0.75] (b);

\draw[thick] (a) edge[bend left=50, markpos=mymark2 at 0.5] (b);
\draw[thick] (a) edge[bend left=50, markpos=edgemark2 at 0.25] (b);
\draw[thick] (a) edge[bend left=50, markpos=anothermark2 at 0.75] (b);

\draw[thick] (a) edge[bend left=13, markpos=mymark3 at 0.5] (b);
\draw[thick] (a) edge[bend left=13, markpos=edgemark3 at 0.33] (b);
\draw[thick] (a) edge[bend left=13, markpos=anothermark3 at 0.75] (b);

\draw[thick] (a) edge[markpos=edgemark4 at 0.33] (b);

\node[circle,fill=black,inner sep=2pt,draw,label = {$p_{m^2}$}] (pm1) at (mymark1) {};
\node[circle,fill=black,inner sep=2pt,draw,label = {$p_{m^2-m}$}] (pm2) at (mymark2) {};
\node[circle,fill=black,inner sep=2pt,draw,label = below:{$p_{m}$}] (p2) at (mymark3) {};

\node[circle,fill=black,inner sep=2pt,draw,label = {$p_{m^2-1}$}] (pm1z) at (anothermark1) {};
\node[circle,fill=black,inner sep=2pt,draw,label = {$p_{m^2-m-1}$}] (pm2z) at (anothermark2) {};
\node[circle,fill=black,inner sep=2pt,draw,label = below:{$p_{m-1}$}] (p2z) at (anothermark3) {};

\node[label = below:{$|e_{0}| = 1$}] (l1) at (edgemark4){};

\node[label = below:{$|e_{1}| = 2^{m + 2} $}] (l2) at (edgemark3){};

\node[label = {[rotate=20]above:$|e_{m-2}| = 2^{m^2 -m + 2} $}] (l2) at (edgemark2){};
\node[label = {[rotate=15]above:$|e_{m-1}| = 2^{m^2+2} $}] (l3) at (edgemark1){};

  \draw[loosely dotted, very thick] (mymark2) -- (mymark3);

\end{tikzpicture}
    \caption{Illustration of a graph $G$ and a point cloud $R$ defined in Example \ref{exa:tall_imbalanced_tree}}
    \label{fig:GraphConstructionOfExample}
\end{figure}

\begin{figure}
    \centering
    \begin{tikzpicture}[align=center, node distance = 1.35cm, scale = 0.5]

	\node (scalem2text) {Level $m^2 + 1$};

	\node[below of =scalem2text] (scalem1text) {Level $m^2 $};
	\node[below of =scalem1text] (scalemh01text) {Level $m^2 - 1$};
	\node[below = 50 pt of scalemh01text] (scalem11text) {Level $m \cdot (m-1) + 1$};

		\node[below = 60 pt of  scalem11text] (scalem1text2) {Level $i \cdot m $};
	\node[below of =scalem1text2] (scalemh01text2) {Level $i \cdot m - 1$};
	\node[below = 50 pt of scalemh01text2] (scalem11text22) {Level $m \cdot (i-1) + 1$};

		\node [blockzL,  right=50pt of scalem2text ] (noderoot) {$r$};

	\node [blockzL,  right=180pt of scalem1text] (node1) {$p_{m^2}$};
	\node [blockzL, right=175pt of scalemh01text] (node2) {$p_{m^2-1}$};
			\node [blockzL, right=160pt of scalem11text] (node4) {$p_{m^2 - m +1}$};

				\node [blockzL,  right=164pt of scalem1text2] (node1d) {$p_{i \cdot m}$};
	\node [blockzL, right=155pt of scalemh01text2] (node2d) {$p_{i \cdot m-1}$};
			\node [blockzL, right=140pt of scalem11text22] (node4d) {$p_{i \cdot m - m +1}$};

		\node[below = 60 pt of scalem11text22] (scalem11text2) {Level  $m $};

	\node[below of =scalem11text2] (scalemm1text) {Level $m-1$};
	\node[below = 50 pt of scalemm1text] (level1) {Level $1$};

			\node [blockzL, right=135pt of scalem11text2] (node5) {$p_{m}$};
				
			\node [blockzL, right=125pt of scalemm1text] (node6) {$p_{m-1}$};
			
			\node [blockzL, right=130pt of level1] (node7) {$p_{1}$};


	  \draw[->] (noderoot) -> (node1);
	  \draw[->,dashed] (noderoot) -> (node1d);
	  \draw[->] (node1) -> (node2);
	  \draw[->, dashed] (node2) -> (node4);
	  
	   \draw[->] (node1d) -> (node2d);
	  \draw[->, dashed] (node2d) -> (node4d);
	  \draw[->, dashed] (noderoot) -> (node5);
	  \draw[->] (node5) -> (node6);
	  \draw[->, dashed] (node6) -> (node7);
	 
	 \draw[loosely dotted, thick] (scalem11text22) -- (scalemh01text2);
	 \draw[loosely dotted, very thick] (scalem11text) -- (scalem1text2);
	  \draw[loosely dotted, thick] (scalemh01text) -- (scalem11text);
	  \draw[loosely dotted, very thick] (scalem11text22) -- (scalem11text2);
	  \draw[loosely dotted, thick] (level1) -- (scalemm1text);

\end{tikzpicture}
    \caption{Illustration of the a compressed cover tree $\T(R)$ defined in Example \ref{exa:tall_imbalanced_tree}}
    \label{fig:bad_cover_tree}
\end{figure}
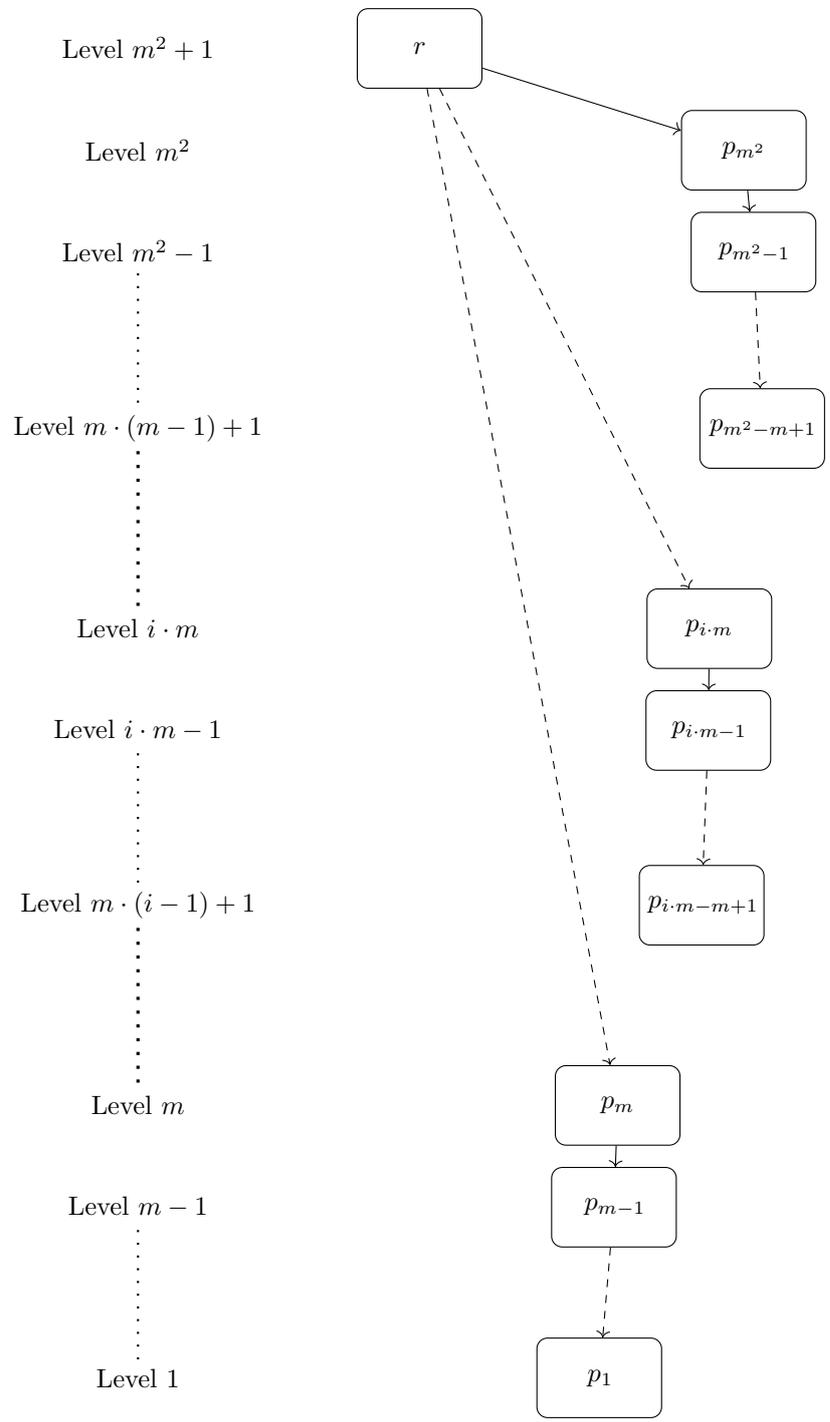

\begin{cexa}
\label{cexa:dualtreecode}
In the notations of  Example\ref{exa:tall_imbalanced_tree}, $m$ is a parameter of $R$.
Build a compressed cover tree $\T(R)$ as in Figure \ref{fig:bad_cover_tree}. 
Set $Q = R$. 
First we show that Algorithm \ref{alg:cover_tree_k-nearest_dt_original} returns the trivial neighbor for every point for $\T(Q)=\T(R)$. 
\medskip

We start the simulation with the query node $r$ on the level $m^2+1$, which has the reference subset $R_{m^2+1} = \{r\}$. 
The query node and the reference set are at the same levels, so we run the query expansions (lines 9-11). 
The node $r$ has $p_{m^2}$ and $r$ as its children.
Hence the algorithm goes into the branches that have $p_{m^2}$ as the query node and into the branches that have $r$ as the query node.
\medskip

 Let us focus on all recursions having $p_{m^2}$ as the query node. In the first recursion involving the node $p_{m^2}$, we have $i = m^2+1, j = m^2$. Thus $j < i$ and we run reference expansions (lines 5-9). 
The node $r$ has two children at the level $m^2$, so $\C(R_i) = \{p_{m^2}, r\}$ . Since
$d(p_{m^2},p_{m^2}) = 0$ and $d(p_{m^2}, r) = 2^{m^2+1}$ on line 7, we have:
$$R_{m^2}  = \{r \in \C(R_i) \mid d(q_j,r) \leq 2^{m^2 + 1} + 2^{m^2+2} \} = \{p_{m^2}, r\}.$$
Similarly, for $i = m^2, j = m^2-1, q_j = p_{m^2}$, we have $\C(R_i) =  \{p_{m^2}, p_{m^2-1}, r\}$ and since $d(p_{m^2 - 1}, p_{m^2}) = 2^{m^2}$ and $d(r,p_{m^2}) = 2^{m^2 + 1}$ we have:
$$R_{m^2-1} =  \{r \in \C(R_i) \mid d(q_j,r) \leq 2^{m^2} + 2^{m^{2} + 1} \} = \{p_{m^2}, p_{m^2-1}\}. $$
For $i = m^2-1, j = m^2-2, q_j = p_{m^2}$, we have $\C(R_i) =  \{p_{m^2}, p_{m^2-1}, p_{m^2-2}\}$.
Since $d(p_{m^2 - 1}, p_{m^2}) = 2^{m^2}$ and $d(p_{m^2-2},p_{m^2}) = 2^{m^2} + 2^{m^2-1}$, we have:
$$R_{m^2-2} =  \{r \in \C(R_i) \mid d(q_j,r) \leq 2^{m^2 - 1} + 2^{m^{2}} \} = \{p_{m^2}, p_{m^2-1}, p_{m^2-2}\}.$$
Finally, for $i = m^2-2, j = m^2-3, q_j = p_{m^2}$, we have $\C(R_i) =  \{p_{m^2}, p_{m^2-1},p_{m^2-2}, p_{m^3-3}\}$ and 
$d(p_{m^2}, p_{m^3-3}) = 2^{m^2} + 2^{m^2-1} + 2^{m^2-2}$. The previous inequalities imply that
$$R_{m^2-3} =  \{r \in \C(R_i) \mid d(q_j,r) \leq 2^{m^2-2} + 2^{m^{2}-1} \} = \{p_{m^2}\}.$$
Since  $R_{m^2 -3} = \{p_{m^2}\}$, the nearest neighbor of $p_{m^2}$ will be chosen to be $p_{m^2}$. 
The same argument can be repeated for all $p_t \in R$. 
It follows that Algorithm \ref{alg:cover_tree_k-nearest_dt_original} finds trivial nearest neighbor for every point $p_t \in R$. 
\bs
\end{cexa}

\begin{exa}
\label{exa:dualtreeprooffaultyconst}
To avoid the issue of finding trivial nearest neighbors in Counterexample~\ref{cexa:dualtreecode}, we will modify Example~\ref{exa:tall_imbalanced_tree}. 
For any integer $m > 100$, let $G$ be a metric graph that has $2$ vertices $r$ and $q$ and $2m-1$ edges $\{e_{0}\} \cup \{e_{1}, ..., e_{m-1},h_{1},..., h_{m-1}\}$.
The edge-lengths are $|e_i| = 2^{i \cdot m+2}$ and $|h_i| = 2^{i \cdot m+2}$  for all $i \in [1,m]$, finally $|e_{0}| = 1$.
\medskip

For every $i \in \{1, ..., m^2\}$, if $i$ is divisible by $m$, we set $q_{i}$ to be the middle point of $e_{i / m}$ and $r_{i}$ to be the middle point of $h_{i / m}$. 
For every other $i$ not divisible by $m$, we define $q_i$ to be the middle point of segment $(q_{i+1}, q)$ and $r_i$ to be the middle point of segment $(r_{i+1}, q)$.
\medskip

Let $d$ be the shortest path metric on the graph $G$. 
Then $d(q_i,r) = d(q_i, q) + 1 = 2^{i+1} + 1$ , $d(q_i, r_j) = 2^{i+1} + 2^{j+1}$ and $d(q,r) = 1$. 
Let $R = \{r, r_{m^2}, r_{m^2-1}, ..., r_1\}$ and let $Q = \{r,  q_{m^2}, q_{m^2-1}, ..., q_1 \}$. 
Let compressed cover trees $\T(Q),\T(R)$ have the same structure as the compressed cover tree $\T(R)$ in Example~\ref{exa:tall_imbalanced_tree}.
\bs
\end{exa}

\begin{rem}
\label{rem:kappa}
\cite[Definition~3.1]{ram2009linear} introduced the degree of bichromaticity $\kappa$ as follows. 
\medskip

"\textbf{Definition 3.1} Let $S$ and $T$ be cover trees built on query set $Q$ and reference set $R$ respectively.
Consider a dual-tree algorithm with the property that the scales of $S$ and $T$ are kept as close as
possible – i.e. the tree with the larger scale is always descended. Then, the degree of bichromaticity
$\kappa$ of the query-reference pair $(Q, R)$ is the maximum number of descends in $S$ between any two
descends in $T$".
\medskip

There are at least two different interpretation of this definition. Our best interpretation is that $\kappa$ the maximal number of levels in $T$ containing at least one node between any two consecutive levels of $S$. However, if $q$ is a leaf node of $S$, but there are still many levels between level of $q$ and $l_{\min}(T)$, it is not clear from the definition if $\kappa$ includes these levels.
\medskip

\cite[page~3284]{curtin2015plug} pointed out that "
Our results are similar to that of Ram et al. (2009a), but those results depend on a
quantity called the constant of bichromaticity, denoted $\kappa$, which has unclear relation to
cover tree imbalance. The dependence on $\kappa$ is given as $c_q^{4\kappa}$ , which is not a good bound,
especially because $\kappa$ may be much greater than 1 in the bichromatic case (where $S_q = S_r$)".
\end{rem}

To keep track of the indices $i,j$ the function call FindAllNN($q_j$, $R_i$) will be expressed as FindAllNN($i,j,q_j, R_i$) in Counterexample \ref{cexa:dualtreecode}.

\begin{cexa}
\label{cexa:dualtreeproof}
We will now show that in addition to the problems in the pseudocode the proof of \cite[Theorem~3.1]{ram2009linear} is incorrect.  Let us consider the following quote from its proof. 
\medskip

"\textbf{Theorem 3.1} Given a reference set $R$ of size $N$ and expansion constant $c_R$, a query set $Q$ of size
$O(N)$ and expansion constant $c_Q$, and bounded degree of bichromaticity $\kappa$ of the $(Q,R)$ pair, the
FindAllNN subroutine of Algorithm 1 computes the nearest neighbor in $R$ of each point in $Q$ in
$O(c^{12}_Rc^{4\kappa}_QN)$ time
\medskip

[\emph{Partial proof:}]
Since at any level of recursion, the size of $R$ [Corresponding to $\C(R_i)$ in Algorithm \ref{alg:cover_tree_k-nearest_dt_original} ] is bounded by $c_R^4\max_i{R_i}$ (width bound), and the
maximum depth of any point in the explicit tree is $O(c^2_R \log(N))$ (depth bound), the number of nodes
encountered in Line 6 is $O(c_R^{6} \max_i |R_i|\log(N))$. Since the traversal down the query tree causes
duplication, and the duplication of any reference node is upper bounded by $c_Q^{4\kappa}$ , Line 6 [corresponds to line 8 in Algorithm \ref{alg:cover_tree_k-nearest_dt_original}] takes at most
$c^{4\kappa}_Qc^6_R\max_i|R_i|\log(N)$ in the whole algorithm. 
"
\medskip

The above arguments claimed the algorithm runs Line 8 at most this number of times:
\begin{equation}
\label{eqa:AuthorClaimingDualTreeKNN}
\#(\text{Line 8}) \leq \max_{p \in R}D(p) \cdot \max_i \C(R_i)  \cdot (\text{number of duplications}).
\end{equation}
Let $X, \T(R), \T(Q), R, Q$ be as in Example \ref{exa:dualtreeprooffaultyconst} for some parameter $m$. We will consider the simulation of Algorithm \ref{alg:cover_tree_k-nearest_dt_original} on pair $(\T(Q),\T(R))$. We note first Lemma \ref{lem:tall_imbalanced_tree_explicit_depth} applied on $\T(R)$ provides $\max_{p \in R}D(p) \leq 2m+1$
Similarly as in \cite[Counterexample~3.7]{elkin2021new}, a contradiction will be achieved by showing that $R_i$ and a set of its children $\C(R_i)$ will have constant size bound on any recursion $(i,j)$ of Algorithm \ref{alg:cover_tree_k-nearest_dt_original}.
\medskip

Since $\T(R)$ contains at most one children on every level $i$ we have $|\C(R_i)| \leq |R_i| + 1$ for any recursion of FindAllNN algorithm. For any $i > m^2$ denote $r_i$ and $q_i$ to be $r$.
 Note first that since $l(q_t) = t$ for any $t \in [1,m^2]$, then $q_t$ is recursed into from FindAllNN($t+1,t+1,p,R_i$), where $p$ is parent node of $q_t$. Therefore it follows that $t \geq i + 1 $ in any stage of the recursion. 
 Let us prove that for any $i \in [1,m^2+1]$ following two claims hold: (1) Function FindAllNN($i$ , $j = i-1$, $q_t$, $R_i$) is called for all $t \geq i + 1$ and (2) We have  $R_i = \{r_{i+1}, r_{i}, r\}$ in this stage of the algorithm. The claim will be proved by induction on $i$. Let us first prove case $i = 2m+1$.  Note that  Algorithm \ref{alg:cover_tree_k-nearest_dt_original} is originally launched from FindAllNN($2m+1,2m+1,r, \{r\}$), 
therefore the first claim holds. Second claim holds trivially since $r_{2m+2} = r$ and $r_{2m+1} = r$.
\medskip
 
 Assume now that the claim holds for some $i$, let us show that the claim will always hold for $i-1$. Assume that FindAllNN($i , j = i-1, q_{t}, R_i)$ was called for some $t \geq i+1$. Since $j = i-1$, we perform a reference expansion (lines 5-9). By line $6$ and induction assumption we have $\C(R_i) = \{r, r_{i+1}, r_{i}, r_{i-1}\}$. Assume first that $q_t = r$.
Recall that for any $u \in [1,m^2]$ we have $d(r,r_{u}) \geq 2^{u+1} $. It follows that
$$R_{i-1} = \{ r' \in \C(R_i) \mid d(r, r') \leq  2^{i} + 2^{i+1}\} = \{r,r_{i},r_{i-1}\}.$$
Let us now consider case $q_t \neq r$. We have $d(r,q_t) = 2^{t+1}$ and $d(q_t, r_{u+1}) = 2^{t+1} + 2^{u+2}$ for any $u \in [1,m^2+1]$. Therefore
$$R_{i-1} = \{ r' \in C_{i-1} \mid d(q_t, r') \leq d(q_t,r) + 2^i + 2^{i+1} \leq 2^{t+1} + 2^{i} + 2^{i+1}\}.$$
It follows that $R_{i-1} = \{r, r_{i}, r_{i-1}\}$. In both cases we proceed to line $8$ where we launch FindAllNN$(i-1 , i-1, q_{t}, R_{i-1})$. After proceeding into the recursion we have $j = i$ and therefore query-expansion (lines 9-11) will be performed.  Note that $q_{t}$ was chosen so that $t \geq i+1$. Since every $q_{t-1}$ is either a child of $r$ or $q_{t}$ it follows that FindAllNN$(i-1 , i-2, q_{t'}, R_{i-1})$ will be called for all $t' \geq t-1 \geq i$. Then condition (2) of the induction claim holds as well. 
\medskip

It remains to show that Algorithm \ref{alg:cover_tree_k-nearest_dt_original} $(q, R_i = \{r\})$ has $O(m^4)$ low bound on the number of times reference expansions (lines 5-9) are performed.
Let $\xi$ be the number of times Algorithm \ref{alg:cover_tree_k-nearest_dt_original} performs 
reference expansions.  For every $q' \in Q$ denote $\xi(q')$ to be the total number of reference expansions performed for $q'$. Recall that any query node $q'$ is introduced in the query expansion (lines 9-11) for parameters $(i = u+1, j = u+1, p, R_i)$, where $p$ is the parent node of $q'$. Since $R_i$ is non empty for all levels $[1,u]$ we have $\xi(q_u) \geq u - 1$ for all $u$. Thus 
$$\xi = \sum_{q' \in Q} \xi(q') \geq \sum^{m^2+1}_{u = 2} u-2 = O(m^4).$$
There are different interpretations for the number of duplications. Note that the query tree $\T(Q)$ has exactly one new child on every level and that trees $\T(Q)$ and $\T(R)$ contain exactly the same levels. By using the definitions the number of duplications should be $2$. However, since there can be other interpretations for the number of duplications, we make a rough estimate that the number of duplication is upper bounded by the number of nodes in query tree $O(m^2)$. By using Inequality (\ref{eqa:AuthorClaimingDualTreeKNN}), we obtain the following contradiction: 
$$O(m^4) = \xi \leq  \max_{p \in R}D(p) \cdot (\max_i \C(R_i)) \cdot (\text{number of duplications}) \leq (2m+1) \cdot 4 \cdot m^2 \leq O(m^3).$$
\end{cexa}

\section{Detailed proofs of all auxiliary lemmas from sections~\ref{sec:cover_tree} and~\ref{sec:distinctive_descendant_set}}
\label{sec:proofs}

All lemmas from sections~\ref{sec:cover_tree} and~\ref{sec:distinctive_descendant_set} are re-stated below (with the same numbers) for convenience.
 
\lemexpansionconstantproperty*
\begin{proof}
The proof follows trivially from Definition~\ref{dfn:expansion_constant}.
\end{proof}

\lemcompressedcovertreedescendantbound*
\begin{proof}
Let $(w_0, ..., w_m)$ be a subpath of the node-to-root path for $w_0 = q$ , $w_{m-1} = u$, $w_m = p$. 
Then $d(w_{i}, w_{i+1}) \leq 2^{l(w_i) + 1}$ for any $i$. 
The first required inequality follows from the triangle inequality below: 
$$    d(p,q) \leq \sum^{m-1}_{j = 0}d(w_j, w_{j+1})  \leq \sum^{m-1}_{j = 0}2^{l(w_j) + 1} \leq   \sum_{t = l_{\min}}^{l(u) + 1}2^{t}\leq  2^{l(u) + 2} $$
Finally, $l(u) \leq l(p) - 1$ implies that $d(p,q) \leq 2^{l(p)+1}$.
\end{proof}

\lempacking*
\begin{proof}
Assume that $d(p,q) > t$ for any point $q \in S$. 
Then $\bar{B}(p, t) \cap S = \emptyset$ and the lemma holds trivially. Otherwise $\bar{B}(p, t) \cap S$ is non-empty. 
By Definition~\ref{dfn:expansion_constant} of a minimized expansion constant, for any $\epsilon > 0$, we can always find a set $A$ such that $S \subseteq A \subseteq X$ and
\begin{ceqn}
\begin{equation}
\label{eqa:dfn_of_exp_constant}
|B(q,2s) \cap A| \leq (c_m(S) + \epsilon) \cdot | B(q,s) \cap A|
\end{equation}
\end{ceqn}
for any $q \in A$ and $s \in \R$. Note that for any $u \in \bar{B}(p,t) \cap S$ we have $\bar{B}(u, \frac{\delta}{2}) \subseteq  \bar{B}(u, t + \frac{\delta}{2})$. 
Therefore, for any point $q \in \bar{B}(p,t) \cap S$, we get
$$\bigcup_{u \in \bar{B}(p, t) \cap S}\bar{B}(u, \frac{\delta}{2}) \subseteq  \bar{B}(p,t + \frac{\delta}{2}) \subseteq \bar{B}(q, 2t +  \frac{\delta}{2})$$
Since all the points of $S$ were separated by $\delta$, we have
\begin{equation*}
\label{eqa:packing_zero}
| \bar{B}(p, t) \cap S| \cdot \min_{u \in \bar{B}(p, t) \cap S}| \bar{B}(u, \frac{\delta}{2}) \cap A|  \leq \sum_{u \in \bar{B}(p, t) \cap S} | \bar{B}(u, \frac{\delta}{2}) \cap A |\leq | \bar{B}(q, 2t + \frac{\delta}{2}) \cap A |
\end{equation*}
In particular, by setting $q = \mathrm{argmin}_{a \in S \cap \bar{B}(p,t)}| \bar{B}(a, \frac{\delta}{2})|$, we get:
\begin{ceqn}
\begin{equation}
\label{eqa:packing_one}
| \bar{B}(p, t) \cap S| \cdot |  \bar{B}(q, \frac{\delta}{2}) \cap A| \leq | \bar{B}(q, 2t + \frac{\delta}{2}) \cap A |
\end{equation}
\end{ceqn}
Inequality~(\ref{eqa:dfn_of_exp_constant}) applied $\mu$ times 
for the radii $s_i = \dfrac{2t + \frac{\delta}{2}}{2^{i}} $, $i = 1,...,\mu$, implies that:
\begin{ceqn}
\begin{equation} 
\label{eqa:packing_two}
 |\bar{B}(q,2t + \frac{\delta}{2}) \cap A| \leq (c_m(S) + \epsilon)^{\mu}|\bar{B}(q, \dfrac{2t + \frac{\delta}{2}}{2^{ \mu}}) \cap A |  \leq  (c_m(S) + \epsilon)^{ \mu}|\bar{B}(q, \frac{\delta}{2}) \cap A|. 
\end{equation}
\end{ceqn}
By combining inequalities (\ref{eqa:packing_one}) and (\ref{eqa:packing_two}), we get
$$| \bar{B}(p,t) \cap S  |\leq \dfrac{|\bar{B}(q, 2t + \frac{\delta}{2})  \cap A |}{|\bar{B}(q, \frac{\delta}{2})  \cap A|} \leq (c_m(S)+\epsilon)^{\mu}.$$
The required inequality is obtained by letting $\epsilon \rightarrow 0$.\end{proof}

\lemcompressedcovertreewidthbound*
\begin{proof}
By the covering condition in Definition~\ref{dfn:cover_tree_compressed} for a compressed cover tree $\T(R)$, any child $q$ of a node $p$ located on the level $i$ has $d(q,p) \leq 2^{i+1}$. Thus the number of children of the node $p$ at the level $i$ at most $|\bar{B}(p,2^{i+1})|$. 
\medskip

The separation condition in Definition~\ref{dfn:cover_tree_compressed} implies that the cover set $C_i$ is $2^{i}$-sparse in $X$. 
We can now apply Lemma \ref{lem:packing}
for $t = 2^{i+1}$ and $\delta = 2^{i}$. Since $4 \cdot \frac{t}{\delta} + 1 \leq 4 \cdot 2 + 1 \leq 2^4$, 
we get $|\bar{B}(q,2^{i+1}) \cap C_i|  \leq (c_m(C_{i}))^4$. 
\medskip

The lemma follows by noting that Lemma \ref{lem:expansion_constant_property} implies $(c_m(C_{i}))^4  \leq (c_m(R))^4 $.
\end{proof}

\lemdepthbound*
\begin{proof}
$|H(\T(R))|\leq l_{\max} - l_{\min}+1$ by Definition~\ref{dfn:depth}. 
We estimate $l_{\max} - l_{\min}$ as follows.
\medskip
 
Let $p \in R$ be a point with $\rad(R) = \max\limits_{q \in R}d(p,q)$. 
Then $R$ is covered by the closed ball $\bar B(p; \rad(R))$.
Hence the cover set $C_i$ at the level $i=\log_2(\rad(R))$ consists of a single point $p$. 
The separation condition in Definition~\ref{dfn:cover_tree_compressed} implies that 
 $l_{\max}\leq \log_2(d_{\max}(R))$.
\medskip

Since any distinct points $p,q \in R$ have a distance $d(p,q)\geq d_{\min}(R)$, the covering condition implies that no new points can enter the cover set $C_i$ at the level $i=[\log_2(d_{\min}(R))]$, so $l_{\min}\geq\log_2(d_{\min}(R))$.
So
$|H(\T(R))| \leq 1+l_{\max} - l_{\min} \leq 
1+\log_2\dfrac{\rad(R)}{d_{\min}(R)}$.
\end{proof}


\lemdistinctivedescendantsprecompute*
\begin{proof}
By Lemma \ref{lem:number_of_explicit_levels} we have $\sum_{p \in R}|\mathcal{E}(p,\T(R))| \leq 2 \cdot |R|.$ Since CountDistinctiveDescendants is called once for every any combination $p \in R$ and $i \in \mathcal{E}(p,\T(R))$ it follows that the time complexity of Algorithm~\ref{alg:cover_tree_distinctive_descendants} is $O(R)$.
\end{proof}

\lemtimelambdapoint*
\begin{proof}
To compute $\lambda = \lambda_k(q,R)$ in Algorithm \ref{alg:lambda}, we need to select $k$ elements by using an ordering from the set $C$.
This selection takes takes at most $|C| \cdot \log(k)$ time by using a binary heap data structure \cite[section~6.5]{Cormen1990}.
\end{proof}

\lemseparation*
\begin{proof}
Without loss of generality assume that $l(p) \geq l(q)$. If $q$ is not a descendant of $p$, the lemma holds trivially due to $\Desc(q) \cap \Desc(p) = \emptyset$. 
\medskip

If $q$ is a descendant of $p$, then $l(q) \leq l(p) - 1$ and therefore $q \in V_i(p)$. 
It follows that
$\Sd_i(p , \T(R)) \cap \Desc(q) = \emptyset$ and therefore
$$\Sd_{i}(p, \T(R)) \cap \Sd_{i}(q, \T(R))   \subseteq \Sd_{i}(p, \T(R)) \cap \Desc(q) = \emptyset.$$ 
\end{proof}

\lemsum*
\begin{proof}
The proof follows from Lemma \ref{lem:separation} by noting that any $p \in V \subseteq C$ has $l(p) \geq i$.
\end{proof}

\lemdistinctivedescendantchildlevel*
\begin{proof}
Let $w \in \Sd_i(p)$ be an arbitrary node satisfying $w \neq p$. 
Let $s$ be the node-to-root path of $w$. 
The inclusion $\Sd_i(p) \subseteq \Desc(p)$ implies that $w \in \Desc(p)$. 
\medskip

Let $a \in \Child(p) \setminus \{p\}$ be a child on the path $s$. 
If $l(a) \geq i$ then $a \in V_i(p)$. 
Note that $w \in \Desc(a)$.
Therefore $w \notin \Sd_i(p)$, which is a contradiction.
Hence $l(a) < i$. 
\end{proof}

\lemchildsetequivalence*
\begin{proof}
Let $a \in \bigcup_{p \in C}\Sd_{i-1}(p, \T(R))$ be an arbitrary node. Then there is $v \in C$ having $ a \in \Sd_{i-1}(v, \T(R))$. 
Let $w \in R_i$ be an ancestor of $v$ that has the lowest level among all ancestors from in $R_i$. 
The node $v$ has always an ancestor in $R_i$.
Indeed, by the choice of $v \in C$, either $v \in R_i$ or $v$ has a parent in $R_i$. Hence $a \in \Desc(w)$. 
Since $w$ has a minimal level among all ancestors of $v$, we conclude that $a \notin  \bigcup_{u \in V_i(w)}\Desc(u)$.
Then
$$a \in \Sd_i(w, \T(R)) \subseteq \bigcup_{p \in R_i}\Sd_i(p , \T(R)).$$
Assume now that $a \in \bigcup\limits_{p \in R_i}\Sd_{i}(p, \T(R))$. 
Then $a \in \Sd_{i}(v, \T(R))$ for some $w \in R_i$.
Let $w$ have no children at the level $i-1$. 
Then $V_i(w) = V_{i-1}(w)$ and 
$$a \in \Sd_{i-1}(w, \T(R)) \subseteq \bigcup_{p \in \C(R_i)}\Sd_{i-1}(p ,\T(R)).$$ 
Assume now that $w$ has children at the level $i-1$.
If there exists $b \in \Child(w)$ for which $a \in \Desc(b)$. 
Since $V_{i-1}(b) = \emptyset$, we conclude that 
$$a \in \Sd_{i-1}(b, \T(R)) \subseteq \bigcup_{p \in C}\Sd_{i-1}(p ,\T(R)).$$ 
To prove the converse inclusion assume that $a \notin \Desc(b)$ for all $b \in \Child(w)$ with $l(b) = i-1$.
Then $a \in \Desc(w)$ and $a \notin \Desc(b')$ for any $b' \in V_i(w)$. Then $a \in \Sd_{i-1}(w, \T(R))$ and the proof finishes:
$$\bigcup_{p \in R_i}\Sd_{i}(p, \T(R))   \subseteq   \bigcup_{p \in \C(R_i)}\Sd_{i-1}(p, \T(R)).  $$
\end{proof}

\lembetapoint*
\begin{proof}
We show that $R$ has a point $\beta$ among the first $k$ nearest neighbors of $q$ such that
$$\beta \in \bigcup_{p \in C}\Sd_i(p, \T(R)) \setminus \bigcup_{p \in N(q, \lambda) \setminus \{\lambda\} }\Sd_i(p, \T(R)).$$
Lemma~\ref{lem:sum} and Definition~\ref{dfn:lambda-point} imply that
$$ | \bigcup_{p \in N(q, \lambda) \setminus \{\lambda\} }\Sd_i(p, \T(R)) |= \sum_{p \in N(q, \lambda) \setminus \{\lambda\} }| \Sd_i(p, \T(R)) | < k.$$ 
Since $\cup_{p \in C}s_i(p, \T(R))$ contains all $k$-nearest neighbors of $q$, a required point $\beta\in R$ exists. 
\medskip
 
Let us now show that $\beta$ satisfies $d(q,\lambda) \leq  d(q,\beta) + 2^{i+1}$.
Let $\gamma \in C \setminus N(q,\lambda) \cup \{\lambda \}$ be an ancestor of $\beta$. 
Since $\gamma \notin N(q,\lambda) \setminus \{\lambda\}$, we get $d(\gamma, q) \geq d(q, \lambda)$.  
The triangle inequality says that $ d(q, \gamma) \leq d(q,\beta) + d(\gamma ,\beta) $.
Finally, Lemma~\ref{lem:compressed_cover_tree_descendant_bound} implies that $d(\gamma, \beta) \leq 2^{i+1}$. 
Then
$$d(q,\lambda) \leq d(q, \gamma) \leq d(q,\beta) + d(\gamma ,\beta) \leq d(q,\beta) +2^{i+1}$$
So $\beta$ is a desired $k$-nearest neighbor satisfying the condition $d(q,\lambda) \leq  d(q,\beta) + 2^{i+1}$.
\end{proof}

\end{document}